\documentclass{tlp}

\usepackage{subcaption}
\usepackage{amssymb}
\usepackage[inline]{enumitem}

\usepackage{stmaryrd} 

\usepackage{hyperref}
\usepackage{cleveref}
\usepackage{xcolor}

\usepackage{xspace}

\usepackage{tikz}
\usetikzlibrary{arrows,shapes,automata,petri,positioning}

\tikzset{
  every picture/.style={very thick},
}

\usepackage{forest}
\usepackage{tikz-qtree}

\usepackage{booktabs}

\newcommand\DLN{Datalog$^\neg$\xspace}

\newcommand{\dng}[1]{{{\sim}#1}}
\newcommand{\gr}[1]{\mathsf{gr}({#1})}
\newcommand{\adg}[1]{\mathsf{adg}({#1})}

\newcommand{\head}[1]{\mathsf{h}({#1})}
\newcommand{\body}[1]{\mathsf{b}({#1})}
\newcommand{\bodyf}[1]{\mathsf{bf}({#1})}
\newcommand{\pbody}[1]{\mathsf{b}^+({#1})}
\newcommand{\nbody}[1]{\mathsf{b}^-({#1})}

\newcommand{\comp}[1]{\mathsf{comp}({#1})}

\newcommand{\lfp}[1]{\mathsf{lfp}({#1})}
\newcommand{\tgsp}[1]{\mathsf{tg}_{sp}({#1})}
\newcommand{\tgst}[1]{\mathsf{tg}_{st}({#1})}
\newcommand{\rhs}[1]{\mathsf{rhs}_{P}({#1})}
\newcommand{\hb}[1]{\mathsf{HB}_{#1}}

\newcommand{\cset}[1]{\mathcal{S}({#1})}

\newcommand{\Top}[1]{T_{P}({#1})}
\newcommand{\Fop}[1]{F_{P}({#1})}

\newcommand{\twod}[1]{\mathbb{B}}
\newcommand{\threed}[1]{\mathbb{B}_{\star}}

\newcommand{\tval}[0]{1}
\newcommand{\fval}[0]{0}
\newcommand{\uval}[0]{\star}

\newcommand{\uniCovStTS}[1]{\langle #1\rangle_{P}^{st}}
\newcommand{\uniCovSuTS}[1]{\langle #1\rangle_{P}^{sp}}

\newcommand{\pname}[0]{\DLN program\xspace}
\newcommand{\pnames}[0]{\DLN programs\xspace}

\newcommand{\upname}[0]{uni-rule \DLN program\xspace}
\newcommand{\upnames}[0]{uni-rule \DLN programs\xspace}

\newcommand{\anbn}[0]{AND-NOT BN\xspace}
\newcommand{\anbns}[0]{AND-NOT BNs\xspace}

\newcommand{\delt}[0]{delocalizing triple\xspace}
\newcommand{\deltfree}[0]{delocalizing-triple-free\xspace}

\newcommand{\mins}[0]{\mathsf{min}_{\leq_s}}

\newcommand{\var}[1]{\mathsf{var}_{{#1}}}

\newcommand{\ig}[1]{\mathsf{G}({#1})}
\newcommand{\syng}[1]{\mathsf{SG}({#1})}

\newcommand{\astg}[1]{\mathsf{astg}({#1})}
\newcommand{\sstg}[1]{\mathsf{sstg}({#1})}

\newcommand{\INneg}[2]{\mathsf{IN}^{-}({#1}, {#2})}
\newcommand{\INpos}[2]{\mathsf{IN}^{+}({#1}, {#2})}
\newcommand{\INall}[2]{\mathsf{IN}({#1}, {#2})}

\newcommand{\sync}[1]{\mathsf{SynC}({#1})}
\newcommand{\onesynper}[1]{\mathcal{SP}({#1})}
\newcommand{\synper}[1]{\mathcal{SP}^{\omega}({#1})}

\newcommand{\translfp}[2]{\Pi_{{#1}}{#2}}

\newcommand{\acnlp}[1]{\text{NLP}}
\newcommand{\acdg}[1]{\text{DG}}
\newcommand{\acstpm}[1]{\text{StPM}}
\newcommand{\acsupm}[1]{\text{SuPM}}
\newcommand{\acstm}[1]{\text{StM}}
\newcommand{\acsum}[1]{\text{SuM}}
\newcommand{\acregm}[1]{\text{RegM}}

\newcommand{\acbn}[0]{\text{BN}\xspace}
\newcommand{\acbns}[0]{\text{BNs}\xspace}

\newcommand{\acstg}[0]{\text{STG}\xspace}

\newcommand{\nlp}[0]{\text{NLP}\xspace}
\newcommand{\nlps}[0]{\text{NLPs}\xspace}

\newcommand{\ifftext}{\text{iff}\xspace}
\newcommand{\wrttext}{\text{w.r.t.}\xspace}
\newcommand{\unitext}{\text{uniqueness}\xspace}
\newcommand{\unitextu}{\text{Uniqueness}\xspace}

\newtheorem{theorem}{Theorem}[section]
\newtheorem{lemma}{Lemma}[section]
\newtheorem{corollary}{Corollary}[section]
\newtheorem{proposition}{Proposition}[section]

\newtheorem{definition}{Definition}[section]
\newtheorem{example}{Example}[section]
\newtheorem{counter-example}{Counter-example}[section]
\newtheorem{remark}{Remark}[section]
\newtheorem{conjecture}{Conjecture}[section]

\hypersetup{
    colorlinks=true,
    linkcolor=blue,
    citecolor=blue,
}

\begin{document}

\lefttitle{Trinh, Benhamou, Soliman and Fages}

\jnlPage{}{xx}

\title{On the Boolean Network Theory of \DLN}

\begin{authgrp}
	\author{\sn{Van-Giang} \gn{Trinh}}
	\affiliation{Inria Saclay, EP Lifeware, Palaiseau, France}
	\author{\sn{Belaid} \gn{Benhamou}}
	\affiliation{LIRICA team, LIS, Aix-Marseille University, Marseille, France}
	\author{\sn{Sylvain} \gn{Soliman}}
	\affiliation{Inria Saclay, EP Lifeware, Palaiseau, France}
	\author{\sn{Fran\c{c}ois} \gn{Fages}}
	\affiliation{Inria Saclay, EP Lifeware, Palaiseau, France}
\end{authgrp}


\maketitle

\begin{abstract}
Datalog$^\neg$ is a central formalism used in a variety of domains ranging from deductive databases and abstract argumentation frameworks to answer set programming. Its model theory is the finite counterpart of the logical semantics developed for normal logic programs, mainly based on the notions of Clark's completion and two-valued or three-valued canonical models including supported, stable, regular and well-founded models. In this paper we establish a formal link between Datalog$^\neg$ and Boolean network theory first introduced for gene regulatory networks. We show that in the absence of odd cycles in a Datalog$^\neg$ program, the regular models coincide with the stable models, which entails the existence of stable models, and in the absence of even cycles, we prove the uniqueness of stable partial models and regular models. This connection also gives new upper bounds on the numbers of stable partial, regular, and stable models of a Datalog$^\neg$ program using the cardinality of a feedback vertex set in its atom dependency graph. Interestingly, our connection to Boolean network theory also points us to the notion of trap spaces. In particular we show the equivalence between subset-minimal stable trap spaces and regular models.
\end{abstract}

\begin{keywords}
	deductive database, Datalog, model-theoretic semantics, dynamics, Boolean network, trap space, attractor, model counting, feedback vertex set
\end{keywords}



\section{Introduction}\label{sec:introduction}

\emph{\DLN} is a central formalism of non-monotonic logic  that plays a vital role in a variety of computational domains, including deductive databases, answer set programming, and abstract argumentation frameworks~\citep{CGT1990,ELS1997,SS1997,Niemela1999,WCG2009,AFGL2012,JN2016,CS2017}.
As a syntactic restriction of normal logic programs to function-free rules and finite Herbrand universes, \DLN benefits from favorable computational properties while retaining expressive non-monotonic logical semantics~\citep{CGT1990,Sato1990,GV1994,Niemela1999,BFG2002}.

The model theory of \DLN is closely tied to the well-established logical semantics developed for normal logic programs~\citep{YY1995}.
These include two-valued and three-valued semantics such as \emph{supported models}, \emph{stable models}, \emph{regular models}, and the \emph{well-founded models}, which are typically defined in terms of Clark's completion and its associated fixpoint or model-theoretic characterizations~\citep{Clark1977,Lloyd1984,Przymusinski1990,P1994,YY1994,SZ1997,ELS1997}.
In addition, some semantics such as the stable, or supported class semantics have been proposed to represent the dynamical behavior of a normal logic program~\citep{BS1992,IS2012}.
Each of these semantics provides different insights into the meaning of a logic program, and understanding their relationships is crucial for both theoretical and practical applications.

In this paper\footnote{This paper is a significantly extended and revised version of the paper presented at ICLP 2024~\citep{TBSF24} which derived first results from Boolean network theory to propositional logic programs. Here we adopt the first-order setting of \DLN programs and present some results for normal logic programs.}, we investigate the semantics of \pnames through a new and perhaps unexpected lens: the theory of \emph{Boolean networks}, a formalism originally introduced by Stuart Kauffman and René Thomas to model gene regulatory networks~\citep{Kauffman1969,Thomas1973}.
Boolean networks have since evolved into a rich mathematical framework used to study the dynamical behavior of discrete systems, leading to a wide range of applications from science to engineering, especially in systems biology~\citep{SKIKK2020,TBS2023}.
Notably, \pnames have been widely applied to modeling and analysis of Boolean networks~\citep{Inoue2011,KBS2015,TBS2023,KBT2023,TBP2024}.

The preliminary link between \pnames and Boolean networks can be traced back to the theoretical work by~\cite{Inoue2011}.
It defined a Boolean network encoding for \pnames, which relies on the notion of Clark’s completion~\citep{Clark1977}, and pointed out that the two-valued models of the Clark's completion of a \pname one-to-one correspond to the fixed points of the encoded Boolean network.
The subsequent work~\citep{IS2012} pointed out that the strict supported classes of a \pname one-to-one correspond to the synchronous attractors of the encoded Boolean network.
However, this line of work did not explore the structural relationships in detail, nor did it examine connections with other higher-level semantics.
Moreover, the discussion remained largely conceptual and did not extend to concrete implications for the analysis of \pnames.

In this work, we establish a formal connection between \DLN and Boolean network theory by 1) mapping the dependency structure of atoms in a \pname~\citep{ABW1988} to the influence graph of a Boolean network~\citep{Richard2019}, and 2) relating the supported or stable partial model semantics~\citep{P1994,SZ1997} and the regular model semantics~\citep{YY1994,SZ1997} to the notion of trap spaces in Boolean networks~\citep{KBS2015,TBS2023}.
This connection allows us to transfer a variety of structural results from Boolean network theory~\citep{RMCT2003,RR2013,Richard2019,SKIKK2020,RT2023} to the analysis of \pnames.

Relating graphical representations of a normal logic program and its semantics is an interesting research direction in theory that also has many useful applications in practice~\citep{Sato1990,Fages1994,CT1999,Linke2001,LZ2004,Costantini2006,Fichte2011,FH2021}.
Historically, the first studies in this direction focused on the existence of a unique stable model in classes of programs with special graphical properties on (positive) atom dependency graphs, including positive programs~\citep{GL1988}, acyclic programs~\citep{AB1991}, and locally stratified programs~\citep{GL1988}.
The work by~\cite{Fages1991} gave a simple characterization of stable models as well-supported models, and then showed that for \emph{tight} programs (i.e.~without non-well-founded positive justifications), the stable models of the program coincide with the two-valued models of its Clark's completion~\citep{Fages1994}.
Being finer-represented but more computationally expensive than atom dependency graphs, several other graphical representations (e.g., cycle and extended dependency graphs, rule graphs, block graphs) were introduced and several improved results were obtained~\citep{DT1996,Linke2001,Costantini2006,CP2011}.
There are some recent studies on atom dependency graphs~\citep{FL2023,TB2024}, but they still focus only on stable models.
In contrast, very few studies~\citep{YY1994,ELS1997} have been made about stable partial models or regular models.

By exploiting the established connection between \pnames and Boolean networks, we derive the following key results \wrttext the graphical analysis and theory of \pnames:

\begin{itemize}
	\item[(1)] \textbf{Coincidence of Semantics under Structural Conditions.} We show that in the absence of \emph{odd cycles} in the atom dependency graph of a \pname, the \emph{regular models} coincide with the \emph{stable models}. This collapse of semantics not only simplifies the model-theoretic landscape but also guarantees the \emph{existence} of stable models in such cases.
	
	\item[(2)] \textbf{\unitextu Results.} In the absence of \emph{even cycles} in the atom dependency graph of a \pname, we prove the \emph{\unitext of stable partial models}, which further implies the \emph{\unitext of regular models}. 
	This result provides a clear structural condition under which the semantics of a program becomes deterministic.
	
	\item[(3)] \textbf{Correction and Refinement of Previous Work.} We revisit a seminal work by~\cite{YY1994}, which claimed similar properties (the coincidence of regular models and stable models under the absence of odd cycles, and the \unitext of regular models under the absence of even cycles) for well-founded stratification normal logic programs.
	We reveal problems in their definition of \emph{well-founded stratification} and prove results in the first-order setting for negative programs only.
	
	\item[(4)] \textbf{Upper Bounds via Feedback Vertex Sets.} We derive several upper bounds on the number of stable, regular, and stable partial models for a \pname. 
	These bounds are expressed in terms of the \emph{cardinality of a feedback vertex set} in the atom dependency graph, providing a structural measure of semantics complexity.
	To the best of our knowledge, these insights are new to the theory of \pnames.
	
	\item[(5)] \textbf{Stronger Graphical Analysis Results.} We obtain several stronger graphical analysis results on an important subclass of \pnames, namely \emph{uni-rule} \pnames~\citep{SS1997,CSAD2015}.
	
	\item[(6)] \textbf{Trap Spaces and Semantics Correspondences.} Our exploration reveals that the Boolean network notion of \emph{trap space}---a set of states closed under the system's dynamics---has meaningful analogues in the context of \DLN.
	We introduce the notions of \emph{supported trap space} and \emph{stable trap space}, prove their basic properties, and relate them to other semantics for \pnames, in particular show that the \emph{subset-minimal stable trap spaces} coincide with the \emph{regular models}.
\end{itemize}

\paragraph{Paper Structure.} The remainder of the paper is structured as follows. 
In~\Cref{sec:preliminaries}, we provide background on normal logic programs, \DLN programs, and Boolean networks. \Cref{sec:connection} presents the formal link between \pnames and BNs.
Section~\ref{sec:graphical-analysis} revisits the work of~\cite{YY1994} and provides the new graphical analysis results exploiting the established connection.
\Cref{sec:Datalog-trap-spaces} introduces the notions of stable and supported trap spaces for \pnames, shows their basic properties, and relates them to other semantics.
Finally,~\Cref{sec:conclusion} concludes the paper.

\section{Preliminaries}\label{sec:preliminaries}

\subsection{Normal Logic Programs}\label{subsec:preliminaries-NLP}

We assume that the reader is familiar with the logic program theory and the stable model semantics~\citep{GL1988}.
Unless specifically stated, \nlp means normal logic program.
In addition, we consider the Boolean domain \(\twod{} = \{\fval, \tval\}\), the three-valued domain \(\threed{} = \{\fval, \tval, \uval\}\), and the logical connectives used in this paper are \(\land\) (conjunction), \(\lor\) (disjunction), \(\neg\) (negation), and \(\leftrightarrow\) (equivalence).

\subsubsection{Syntax}

We consider a first-order language built over an infinite alphabet of variables,
and finite alphabets of constant, function and predicate symbols.
The set of first-order \emph{terms} is the least set containing variables, constants and closed by application of function symbols.
An \emph{atom} is a formula of the form \(p(t_1, \dots, t_k)\) where \(p\) is a predicate symbol and \(t_i\) are terms.
A \emph{normal logic program} \(P\) is a \emph{finite} set of \emph{rules} of the form
\[p \gets p_1, \dots, p_m, \dng{p_{m + 1}}, \dots, \dng{p_{k}}\] where \(p\) and \(p_i\) are atoms, \(k \geq m \geq 0\),
and \(\sim\) is a symbol for default negation.
A fact is a rule with \(k = 0\).
For any rule \(r\) of the above form, \(\head{r} = p\) is called the \emph{head} of \(r\), \(\pbody{r} = \{p_1, \dots, p_m\}\) is called the \emph{positive body} of \(r\), \(\nbody{r} = \{p_{m + 1}, \dots, p_{k}\}\) is called the \emph{negative body} of \(r\), and \(\body{r} = \pbody{r} \cup \nbody{r}\) is called the \emph{body} of \(r\).
For convenience, we denote \(\bodyf{r} = \bigwedge_{v \in \pbody{r}}v \wedge \bigwedge_{v \in \nbody{r}}\neg v\) as the \emph{body formula} of \(r\); if \(k = 0\), then \(\bodyf{r} = \tval\).
If \(\nbody{r} = \emptyset, \forall r \in P\), then \(P\) is called a \emph{positive} \nlp.
If \(\pbody{r} = \emptyset, \forall r \in P\), then \(P\) is called a \emph{negative} \nlp.

A term, an atom or an \nlp is \emph{ground} if it contains no variable.
The \emph{Herbrand base} of an \nlp \(P\) (denoted by \(\hb{P}\)) is the set of ground atoms formed over the alphabet of \(P\).
It is finite in absence of function symbols, which is the case of \emph{\DLN} programs~\citep{CGT1990}.
The \emph{ground instantiation} of an \nlp \(P\) (denoted by \(\gr{P}\)) is the set of the ground instances of all the rules in \(P\).
An \nlp \(P\) is called \emph{uni-rule} if for every atom \(a \in \hb{P}\), \(\gr{P}\) contains at most one rule whose head is \(a\)~\citep{SS1997}.

\subsubsection{Atom Dependency Graph}

We first recall some basic concepts in the graph theory.

\begin{definition}\label{def:signed-directed-graph}
	A signed directed graph \(G\) on \(\{\oplus, \ominus\}\) is defined as a tuple \((V(G), E(G))\) where \(V(G)\) is a (possibly infinite) set vertices and \(E(G)\) is a (possibly infinite) set of arcs of the form \((uv, s)\) where \(u, v \in V(G)\) (unnecessarily distinct) and either \(s = \oplus\) or \(s = \ominus\).
	The \emph{in-degree} of a vertex \(v\) is defined as the number of arcs ending at \(v\).
	Then the \emph{minimum in-degree} of \(G\) is the minimum value of the in-degrees of all vertices in \(V(G)\).
	An arc \((uv, \oplus)\) (resp.\ \((uv, \ominus)\)) is called a positive (resp.\ negative) arc, and can be written as \(u \xrightarrow{\oplus} v\) or \(v \xleftarrow{\oplus} u\) (resp.\ \(u \xrightarrow{\ominus} v\) or \(v \xleftarrow{\ominus} u\)).
	The graph \(G\) is \emph{strongly connected} if there is always a path between two any vertices of \(G\).
	It is \emph{sign-definite} if there cannot be two arcs with different signs between two different vertices.
\end{definition}

\begin{definition}\label{def:sub-graph}
	A signed directed graph \(G\) is called a \emph{sub-graph} of a signed directed graph \(H\) \ifftext \(V(G) \subseteq V(H)\) and \(E(G) \subseteq E(H)\).
\end{definition}

\begin{definition}\label{def:path-cycle}
	Consider a signed directed graph \(G\).
	A path of \(G\) is defined as a (possibly infinite) sequence of vertices \(v_0v_1v_2\dots\) such that \((v_iv_{i + 1}, \oplus) \in E(G)\) or \((v_iv_{i + 1}, \ominus) \in E(G)\), for all \(i \geq 0\) and except the starting and ending vertices, any other pair of two vertices are distinct.
	When the starting and ending vertices coincide, the path is called a \emph{cycle}.
	In this case, the number of vertices in this cycle is finite.
	A cycle is called \emph{even} (resp.\ \emph{odd}) if its number of negative arcs is even (resp.\ odd).
	In addition, a cycle \(C\) of \(G\) can be seen as a sub-graph of \(G\) such that \(V(C)\) is the set of vertices of \(C\) and \(E(C)\) is the set of arcs of \(C\).
	It is easy to derive that \(C\) is strongly connected and the in-degree of each vertex in \(V(C)\) is 1 within \(C\).
\end{definition}

\begin{definition}\label{def:adding-substracting-operators}
	Given a signed directed graph \(G\) and an arc \((uv, \epsilon)\) with \(u, v \in V(G)\), \(G + (uv, \epsilon)\) is a signed directed graph \((V(G), E(G) \cup \{(uv, \epsilon)\})\) and \(G - (uv, \epsilon)\) is a signed directed graph \((V(G), E(G) \setminus \{(uv, \epsilon)\})\).
\end{definition}

Then we show the formal definition of \emph{atom dependency graph}, the most prominent graphical representation of an \nlp~\citep{Fages1994}.

\begin{definition}\label{def:NLP-atom-dependency-graph}
	Consider an \nlp \(P\).
	The atom dependency graph of \(P\), denoted by \(\adg{P}\), is a signed directed graph over \(\{\oplus, \ominus\}\) defined as follows:
	\begin{itemize}
		\item \(V(\adg{P}) = HB_P\)
		\item \((uv, \oplus) \in E(\adg{P})\) \ifftext there is a rule \(r \in \gr{P}\) such that \(v = \head{r}\) and \(u \in \pbody{r}\)
		\item \((uv, \ominus) \in E(\adg{P})\) \ifftext there is a rule \(r \in \gr{P}\) such that \(v = \head{r}\) and \(u \in \nbody{r}\)
	\end{itemize}
\end{definition}

We next recall the notion of tightness, which plays an important role in characterizing certain semantics properties of \nlps.
An \nlp is called \emph{tight} if its atom dependency graph contains no infinite descending chain \(v_0 \xleftarrow{\oplus} v_1 \xleftarrow{\oplus} v_2 \xleftarrow{\oplus} \dots\) of only positive arcs~\citep{Fages1994}.
Note that for \pnames, the atom dependency graph is always a finite graph; hence a \pname is tight if its atom dependency graph contains no cycle of only positive arcs~\citep{DHW2014}.

\subsubsection{Semantics of Negation}

In this section, we recall several semantics of \nlps in presence of default negation.

\paragraph{Stable and Supported Partial Models.}

A \emph{three-valued interpretation} \(I\) of an \nlp \(P\) is a mapping \(I \colon \hb{P} \to \threed{}\).
If \(I(a) \neq \uval, \forall a \in \hb{P}\), then \(I\) is a \emph{two-valued interpretation} of \(P\).
Usually, a two-valued interpretation is used interchangeably with the set of ground atoms that are true in this interpretation.
A three-valued interpretation \(I\) characterizes  the set of two-valued interpretations denoted by \(\cset{I}\) as \(\cset{I} = \{J | J \in 2^{\hb{P}}, J(a) = I(a), \forall a \in \hb{P}, I(a) \neq \uval\}\).
For example, if \(I = \{p = \tval, q = \fval, r = \uval\}\), then \(\cset{I} = \{\{p\}, \{p, r\}\}\).

We consider two orders on three-valued interpretations.
The truth order \(\leq_{t}\) is given by \(\fval <_{t} \uval <_{t} \tval\).
Then, \(I_1 \leq_t I_2\) \ifftext \(I_1(a) \leq_{t} I_2(a), \forall a \in \hb{P}\).
The order \(\leq_{s}\) is given by \(\fval <_{s} \uval\), \(\tval <_{s} \uval\), and it contains no other relation.
Then, \(I_1 \leq_{s} I_2\) \ifftext \(I_1(a) \leq_{s} I_2(a), \forall a \in \hb{P}\). 
In addition, \(I_1 \leq_{s} I_2\) \ifftext \(\cset{I_1} \subseteq \cset{I_2}\), i.e., \(\leq_{s}\) is identical to the subset partial order.

Let \(e\) be a propositional formula on \(\hb{P}\).
Then the valuation of \(e\) under a three-valued interpretation \(I\) (denoted by \(I(e)\)) following the three-valued logic is defined recursively as follows:
\begin{align*}
	I(e) = \begin{cases}
		e &\text{if } e \in \threed{} \\
		I(e) &\text{if } e = a, a \in \hb{P} \\
		\neg I(e_1) &\text{if } e = \neg e_1\\
		\text{min}_{\leq_t}(I(e_1), I(e_2)) &\text{if } e = e_1 \land e_2\\
		\text{max}_{\leq_t}(I(e_1), I(e_2)) &\text{if } e = e_1 \lor e_2
	\end{cases}
\end{align*} where \(\neg \tval = \fval, \neg \fval = \tval, \neg \uval = \uval\), and \(\text{min}_{\leq_t}\) (resp.\ \(\text{max}_{\leq_t}\)) is the function to get the minimum (resp.\ maximum) value of two values \wrttext the order \(\leq_t\).
We say three-valued interpretation \(I\) is a \emph{three-valued model} of an \nlp \(P\) \ifftext for each rule \(r \in \gr{P}\), \(I(\bodyf{r}) \leq_{t} I(\head{r})\).

Let \(I\) be a three-valued interpretation of an \nlp \(P\).
We build the \emph{reduct} of \(P\) \wrttext \(I\) (denoted by \(P^I\)) as follows.
\begin{itemize}
	\item Remove any rule \(a \leftarrow a_1, \dots, a_m, \dng{b_1}, \dots, \dng{b_k} \in \gr{P}\) if \(I(b_i) = \tval\) for some \(1 \leq i \leq k\).
	\item Afterwards, remove any occurrence of \(\dng{b_i}\) from \(\gr{P}\) such that \(I(b_i) = \fval\).
	\item Then, replace any occurrence of \(\dng{b_i}\) left by a special atom \textbf{u} such that \(\textbf{u} \not \in \hb{P}\) and it always receives the value \(\uval\).
\end{itemize} The reduct \(P^I\) is positive and has a unique \(\leq_{t}\)-least three-valued model.
See~\cite{Przymusinski1990} for the method for computing this \(\leq_{t}\)-least three-valued model of \(P^I\).
Then \(I\) is a \emph{stable partial model} of \(P\) \ifftext \(I\) is equal to the \(\leq_{t}\)-least three-valued model of \(P^I\).
A stable partial model \(I\) is a \emph{regular model} if it is \(\leq_s\)-minimal~\citep{ELS1997}.
A regular model is non-trivial if it is not two-valued.
In the two-valued setting, \(P^I\) is identical to the reduct defined by~\cite{GL1988}.
A two-valued interpretation \(I\) is a \emph{stable model} of \(P\) \ifftext \(I\) is equal to the \(\leq_{t}\)-least two-valued model of \(P^I\).
It is easy to derive that a stable model is a two-valued regular model, as well as a two-valued stable partial model~\citep{AD1995}.

To further relate the above two-valued or three-valued stable semantics to classical logic, we now recall the notion of Clark's completion and the associated concept of two-valued or three-valued supported models.
The (propositional) \emph{Clark's completion} of an \nlp \(P\) (denoted by \(\comp{P}\)) consists of the following equivalences: \(p \leftrightarrow \bigvee_{r \in \gr{P}, \head{r} = p}\bodyf{r}\), for each \(p \in \hb{P}\); if there is no rule whose head is \(p\), then the equivalence is \(p \leftrightarrow \fval\).
Let \(\rhs{a}\) denote the right hand side of the equivalene of ground atom \(a \in \hb{P}\) in \(\comp{P}\).
A three-valued interpretation \(I\) is a three-valued model of \(\comp{P}\) \ifftext for every \(a \in \hb{P}\), \(I(a) = I(\rhs{a})\).
In this work, we define a \emph{supported partial model} of \(P\) as a three-valued model of \(\comp{P}\), and a \emph{supported model} of \(P\) as a two-valued model of \(\comp{P}\).
Of course, a supported model is a (two-valued) partial supported model.
It has been pointed out that a stable partial model (resp.\ stable model) is a supported partial model (resp.\ supported model), but the reverse may be not true~\citep{DHW2014}.

\paragraph{Stable and Supported Classes.} 

We first recall two semantic operators that capture key aspects of the (two-valued) stable and supported model semantics.
Let \(P\) be an NLP and let \(I\) be any two-valued interpretation of \(P\).
We have that \(P^I\) is positive, and has a unique \(\leq_{t}\)-least two-valued model (say \(J\)).
We define the operator \(F_P\) as \(\Fop{I} = J\).
In contrast, we define the operator \(T_P\) as \(\Top{I} = J\) where \(J\) is a two-valued interpretation such that for every \(a \in \hb{P}\), \(J(a) = I(\rhs{a})\).

\begin{definition}[\cite{BS1992}]\label{def:NLP-StC}
	A non-empty set \(S\) of two-valued interpretations is a \emph{stable class} of an \nlp \(P\) \ifftext it holds that \(S = \{\Fop{I} \mid I \in S\}\).
\end{definition}

\begin{definition}[\cite{IS2012}]\label{def:NLP-SuC}
	A non-empty set \(S\) of two-valued interpretations is a \emph{supported class} of an \nlp \(P\) \ifftext it holds that \(S = \{\Top{I} \mid I \in S\}\).
\end{definition}

Trivially, a stable (resp.\ supported) class of size 1 is equivalent to a stable (resp.\ supported) model.
It is well-known that a stable model is a supported model.
However, a stable class may not be a supported class.
A stable (resp.\ supported) class \(S\) of \(P\) is \emph{strict} \ifftext no proper subset of \(S\) is a stable (resp.\ supported) class of \(P\).
An \nlp \(P\) always has at least one stable class~\citep{BS1992}.
If \(P\) is a \pname, it has at least one strict stable class~\citep{BS1992}.
Similarly, a \pname has at least one supported class, as well as at least one strict supported class~\citep{IS2012}.

To better understand the dynamics underlying stable and supported classes, we now turn to a graph-theoretic characterization based on the stable and supported semantics operators.
The \emph{stable} (resp.\ \emph{supported}) \emph{transition graph} of \(P\) is a directed graph (denoted by \(\tgst{P}\) (resp.\ \(\tgsp{P}\))) on the set of all possible two-valued interpretations of \(P\) such that \((I, J)\) is an arc of \(\tgst{P}\) (resp.\ \(\tgsp{P}\)) \ifftext \(J = \Fop{I}\) (resp.\ \(J = \Top{I}\)).
The stable or supported transition graph exhibits the dynamical behavior of an \nlp~\citep{IS2012}.
It has been pointed that the strict stable (resp.\ supported) classes of \(P\) coincide with the simple cycles of \(\tgst{P}\) (resp.\ \(\tgsp{P}\))~\citep{IS2012}.

\subsubsection{Least Fixpoint}

We shall use the fixpoint semantics of normal logic programs~\citep{DK1989} to prove many new results in the next sections.
To be self-contained, we briefly recall the definition of the \emph{least fixpoint} of an \nlp \(P\) as follows.
Let \(r\) be the rule \(p \leftarrow \dng{p_1}, \dots, \dng{p_k}, q_1, \dots, q_j\) in \(\gr{P}\) and let \(r_i\) be the rules \(q_i \leftarrow \dng{q^1_i}, \dots, \dng{q^{l_i}_i}\) in \(\gr{P}\) where \(1 \leq i \leq j\) and \(l_i \geq 0\).
Then \(\translfp{r}{(\{r_1, \dots, r_j\})}\) is the following rule \[p \leftarrow \dng{p_1}, \dots, \dng{p_k}, \dng{q_1^1}, \dots, \dng{q_1^{l_1}}, \dots, \dng{q_j^1}, \dots, \dng{q_j^{l_j}},\] which means that each atom \(q_j\) in the positive body of \(r\) is substituted with the body of the rule \(r_j\).
Then \(\translfp{P}{}\) is the transformation on negative normal logic programs: \[\translfp{P}{(Q)} = \{\translfp{r}{(\{r_1, \dots, r_j\})} | r \in \gr{P}, r_i \in Q, 1 \leq i \leq j\}.\]
Let \(\lfp{P}_i = (\translfp{P}{(\emptyset}))^i = \translfp{P}{(\translfp{P}{(\dots \translfp{P}{(\emptyset)}))}}\), then \(\lfp{P} = \bigcup_{i \geq 1}\lfp{P}_i\) is the least fixpoint of \(P\).
In the case of \pnames, \(\lfp{P}\) is finite and also negative~\citep{DK1989}.

Finally, we now illustrate the key notions introduced so far---such as stable and supported partial models, regular models, tightness, least fixpoint, and transition graphs---through a concrete example of a \pname.

\begin{example}\label{exam:Datalog-all}
	Consider \pname \(P\) (taken from Example 2.1 of~\cite{IS2012}) where \(P = \{p \leftarrow \dng{q}; q \leftarrow \dng{p}; r \leftarrow q\}\).
	Herein, we use ';' to separate program rules. 
	Note that \(\gr{P} = P\).
	The least fixpoint of \(P\) is \(\lfp{P} = \{p \leftarrow \dng{q}; q \leftarrow \dng{p}; r \leftarrow \dng{p}\}\).
	Figures~\ref{fig:exam-Datalog-adg-tgst-tgsp}~(a),~(b), and~(c) show the atom dependency graph, the stable transition graph, and the supported transition graph of \(P\), respectively.
	The program \(P\) is tight, since every cycle of \(\adg{P}\) has a negative arc.
	Consider five three-valued interpretations of \(P\): \(I_1 = \{p = \tval, q = \fval, r = \uval\}\), \(I_2 = \{p = \fval, q = \tval, r = \uval\}\), \(I_3 = \{p = \uval, q = \uval, r = \uval\}\), \(I_4 = \{p = \tval, q = \fval, r = \fval\}\),  and \(I_5 = \{p = \fval, q = \tval, r = \tval\}\).
	Among them, only \(I_3\), \(I_4\), and \(I_5\) are stable (also supported) partial models of \(P\).
	The program \(P\) has two regular models (\(I_4\) and \(I_5\)) that are also stable (also supported) models of \(P\).
	It has three strict stable classes that correspond to three cycles of the stable transition graph of \(P\) (see~\Cref{fig:exam-Datalog-adg-tgst-tgsp}~(b)): \(C_1 = \{p\} \rightarrow \{p\}\), \(C_2 = \{q, r\} \rightarrow \{q, r\}\), and \(C_3 = \emptyset \rightarrow \{p, q, r\} \rightarrow \emptyset\).
	It has three strict supported classes that correspond to three cycles of the supported transition graph of \(P\) (see~\Cref{fig:exam-Datalog-adg-tgst-tgsp}~(c)): \(C_1 = \{p\} \rightarrow \{p\}\), \(C_2 = \{q, r\} \rightarrow \{q, r\}\), and \(C_4 = \{p, q\} \rightarrow \{r\} \rightarrow \{p, q\}\).
	Note that \(C_1\) and \(C_2\) are cycles of size 1 and correspond to stable (also supported) models of \(P\).
\end{example}

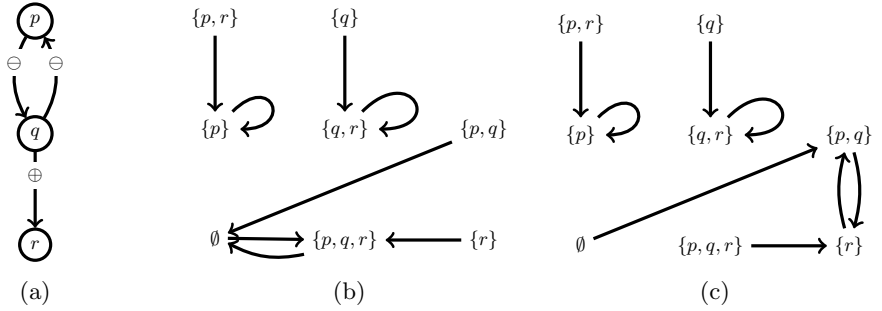
\begin{figure}[!ht]
	\centering
	\begin{subfigure}{0.25\textwidth}
		\centering
		\begin{tikzpicture}[node distance=1cm and 1cm, every node/.style={scale=0.8}]
			\node[circle, draw] (p) [] {$p$};
			\node[circle, draw] (q) [below=of p] {$q$};
			\node[circle, draw] (r) [below=of q] {$r$};
			
			\draw[->] (q) edge [bend right=30] node [midway, above, fill=white] {$\ominus$} (p);
			\draw[->] (p) edge [bend right=30] node [midway, above, fill=white] {$\ominus$} (q);
			\draw[->] (q) edge [] node [midway, above, fill=white] {$\oplus$} (r);
		\end{tikzpicture}
		\caption{}
	\end{subfigure}
	\begin{subfigure}{0.35\textwidth}
		\centering
		\begin{tikzpicture}[node distance=1cm and 1cm, every node/.style={scale=0.8}]
			\node[] (pr) [] {$\{p, r\}$};
			\node[] (p) [below=of pr] {$\{p\}$};
			\node[] (q) [right=of pr] {$\{q\}$};
			\node[] (qr) [below=of q] {$\{q, r\}$};
			\node[] (e) [below=of p] {$\emptyset$};
			\node[] (pqr) [below=of qr] {$\{p, q, r\}$};
			\node[] (r) [right=of pqr] {$\{r\}$};
			\node[] (pq) [right=of qr] {$\{p, q\}$};
			
			\draw[->] (pr) edge [] (p);
			\draw[->] (q) edge [] (qr);
			
			\draw[->] (pq) edge [] (e);
			\draw[->] (r) edge [] (pqr);
			
			\draw[->] (e) edge [] (pqr);
			\draw[->] (pqr) edge [bend left=20] (e);
			
			\draw[->] (p) edge [in=0, out=45, loop] (p);
			\draw[->] (qr) edge [in=0, out=45, loop] (qr);
		\end{tikzpicture}
		\caption{}
	\end{subfigure}
	\begin{subfigure}{0.35\textwidth}  
		\centering
		\begin{tikzpicture}[node distance=1cm and 1cm, every node/.style={scale=0.8}]
			\node[] (pr) [] {$\{p, r\}$};
			\node[] (p) [below=of pr] {$\{p\}$};
			\node[] (q) [right=of pr] {$\{q\}$};
			\node[] (qr) [below=of q] {$\{q, r\}$};
			\node[] (e) [below=of p] {$\emptyset$};
			\node[] (pqr) [below=of qr] {$\{p, q, r\}$};
			\node[] (r) [right=of pqr] {$\{r\}$};
			\node[] (pq) [right=of qr] {$\{p, q\}$};
			
			\draw[->] (pr) edge [] (p);
			\draw[->] (q) edge [] (qr);
			\draw[->] (e) edge [] (pq);
			\draw[->] (pqr) edge [] (r);
			\draw[->] (pq) edge [bend left=15] (r);
			\draw[->] (r) edge [bend left=15] (pq);
			
			\draw[->] (p) edge [in=0, out=45, loop] (p);
			\draw[->] (qr) edge [in=0, out=45, loop] (qr);
		\end{tikzpicture}
		\caption{}
	\end{subfigure}
	\caption{(a) Atom dependency graph \(\adg{P}\), (b) stable transition graph \(\tgst{P}\), and (c) supported transition graph \(\tgsp{P}\) of \pname \(P\) of~\Cref{exam:Datalog-all}.}\label{fig:exam-Datalog-adg-tgst-tgsp}
\end{figure}

\subsection{Boolean Networks}\label{subsec:preliminaries-BN}

\subsubsection{Definition}

A Boolean Network (\acbn) \(f\) is a \emph{finite} set of Boolean functions on a finite set of Boolean variables denoted by \(\var{f}\).
Each variable \(v \in \var{f}\) is associated with a Boolean function \(f_v \colon \mathbb{B}^{|\var{f}|} \rightarrow \mathbb{B}\).
Function \(f_v\) is called \emph{constant} if it is always evaluated to either 0 or 1 regardless of its arguments.
A state \(s\) of \(f\) is a Boolean vector \(s \in \twod{}^{|\var{f}|}\).
State \(s\) can be also seen as a mapping \(s \colon \var{f} \to \twod{}\) that assigns either 0 (inactive) or 1 (active) to each variable.
We can write \(s_v\) instead of \(s(v)\) for short.
For convenience, we write a state simply as a string of values of variables in this state, for instance, we write \(010\) instead of \((0, 1, 0)\).

Among various types of \acbns, one particularly simple yet expressive class is the so-called \emph{AND-NOT} \acbns~\citep{RR2013}, which we introduce next.
A \acbn \(f\) is called an \emph{AND-NOT} \acbn if for every \(v \in \var{f}\), \(f_v\) is 0, 1, or a conjunction of literals (i.e., variables or their negations connected by logical-and rules).
We assume that \(f_v\) does not contain two literals of the same variable.

\subsubsection{Influence Graph}

We now introduce the concept of \emph{Influence Graph} (IG), a graphical representation that captures the local effects of variables on one another within a \acbn~\citep{Richard2019}.
This structure provides an intuitive way to visualize how changes in the value of one variable can influence the outcome of another, either positively or negatively. 
The following formal definition makes this precise.

Let \(x\) be a state of \(f\).
We use \(x[v \leftarrow a]\) to denote the state \(y\) so that \(y_v = a\) and \(y_u = x_u, \forall u \in \var{f} \setminus \{v\}\) where \(a \in \twod{}\).
The influence graph of \(f\) (denoted by \(\ig{f} \)) is a signed directed graph over the set of signs \(\{\oplus, \ominus\}\) where \(V(\ig{f}) = \var{f}\), \((uv, \oplus) \in E(\ig{f})\) (i.e., \(u\) positively affects the value of \(f_v\)) \ifftext there is a state \(x\) such that \(f_v(x[u \leftarrow 0]) < f_v(x[u \leftarrow 1])\), and \((uv, \ominus) \in E(\ig{f})\) (i.e., \(u\) negatively affects the value of \(f_v\)) \ifftext there is a state \(x\) such that \(f_v(x[u \leftarrow 0]) > f_v(x[u \leftarrow 1])\).

To facilitate the structural analysis of the influence graph, we introduce the concept of \emph{feedback vertex set}, which is instrumental in characterizing the behavior of cycles within the graph.
A feedback vertex set allows us to control such feedback by identifying a subset of variables whose removal breaks all cycles of a certain type.
Specifically, an even (resp.\ odd) feedback vertex set of a signed directed graph \(G\) is a (possibly empty) subset of \(V(G)\) that intersects every even (resp.\ odd) cycle of \(G\).
These notions will be particularly useful in later sections when we study the graphical properties of models in \pnames.

\subsubsection{Dynamics and Attractors}

Given a \acbn \(f\), an update scheme specifies the way the variables in \(\var{f}\) update their states at each time step.
There are two major types of update schemes~\citep{SKIKK2020}: synchronous (in which all variables update simultaneously) and asynchronous (in which a single variable is non-deterministically chosen for updating).
Under the update scheme employed, the dynamics of the \acbn are represented by a directed graph, called the \emph{State Transition Graph} (STG).
We use \(\sstg{f}\) (resp.\ \(\astg{f}\)) to denote the synchronous (resp.\ asynchronous) STG of \(f\), and formally define them as follows.

\begin{definition}\label{def:BN-SSTG}
	Given a \acbn \(f\),
	the \emph{synchronous state transition graph} of \(f\) (denoted by \(\sstg{f}\)) is given as: \(V(\sstg{f}) = \twod{}^{|\var{f}|}\) (the set of possible states of \(f\)) and \((x, y) \in E(\sstg{f})\) \ifftext \(y_v = f_v(x)\) for every \(v \in \var{f}\).
\end{definition}

\begin{definition}\label{def:BN-ASTG}
	Given a \acbn \(f\), the \emph{asynchronous state transition graph} of \(f\) (denoted by \(\astg{f}\)) is given as: \(V(\astg{f}) = \twod{}^{|\var{f}|}\) (the set of possible states of \(f\)) and \((x, y) \in E(\astg{f})\) \ifftext there is a variable \(v \in \var{f}\) such that \(y(v) = f_v(x) \neq x(v)\) and \(y(u) = x(u)\) for all \(u \in \var{f} \setminus \{v\}\).
\end{definition}

An \emph{attractor} of a \acbn is defined as a subset-minimal \emph{trap set}, which depends on the employed update scheme.
Equivalently, an attractor is a terminal strongly connected component of the STG corresponding to the employed update scheme~\citep{R2009}.
To ease the statement of our results, we define the concepts of trap set and attractor for directed graphs in general.

\begin{definition}\label{def:digraph-attractor}
	Given a directed graph \(G\),
	a trap set \(A\) of \(G\) is a non-empty subset of \(V(G)\) such that there is no arc in \(G\) going out of \(A\) (formally, there do not exist two vertices \(x \in A\) and \(y \in V(G) \setminus A\) such that \((x, y) \in E(G)\)).
	An \emph{attractor} of \(G\) is defined as a \(\subseteq\)-minimal trap set of \(G\).
	An attractor is called a \emph{fixed point} if it consists of only one vertex, and a \emph{cyclic attractor} otherwise.
	Equivalently, \(A\) is an attractor of \(G\) \ifftext \(A\) is a terminal strongly connected component of \(G\).
\end{definition}

\begin{definition}\label{def:BN-attractor}
	Given a \acbn \(f\), a synchronous (resp.\ an asynchronous) attractor of \(f\) is defined as an attractor of \(\sstg{f}\) (resp.\ \(\astg{f}\)).
\end{definition}

Regarding the synchronous update scheme, each vertex of \(\sstg{f}\) has exactly one out-going arc (possibly a self-arc).
Hence, an attractor of \(\sstg{f}\) is equivalent to a simple cycle of \(\sstg{f}\).
Regarding the asynchronous update scheme, the graph \(\astg{f}\) may contain multiple (up to \(|\var{f}|\)) arcs out from a vertex.
Hence, an attractor of \(\astg{f}\) may be a terminal strongly connected component comprising multiple overlapping cycles.

\subsubsection{Trap Spaces}

To better understand the long-term behavior of a \acbn, it is useful to reason not just about individual states, but also about collections of states that share common properties.
One such useful abstraction is the notion of \emph{trap space}~\citep{KBS2015}, which generalizes the concept of trap sets by allowing partial specification of variable values.
Trap spaces serve as a compact representation of dynamically invariant regions of the state space, within which the system cannot escape once entered.
This notion is particularly attractive because it is independent of the specific update scheme (e.g., synchronous or asynchronous) and can be analyzed purely from the network's structure~\citep{KBS2015,TBS2023}.
We begin by formally defining trap spaces and explaining their relationship with sub-spaces and attractors in the context of \acbns, followed by a concrete example.

A \emph{sub-space} \(m\) of a \acbn \(f\) is a mapping \(m \colon \var{f} \to \threed{}\).
A sub-space \(m\) represents the set of all states \(s\) (denoted by \(\cset{m}\)) such that \(s(v) = m(v), \forall v \in \var{f}, m(v) \neq \uval\).
It is easy to see that the notion of sub-space (resp.\ state) in \acbns is identical to the notion of three-valued interpretation (resp.\ two-valued interpretation) in normal logic programs.
For convenience, we write a sub-space simply as a string of values of variables in this sub-space following an order of variables, for instance, we write \(01\star\) instead of \(\{p = \fval, q = \tval, r = \uval\}\) \wrttext the variable order \((p, q, r)\).
If a sub-space is also a trap set, it is a \emph{trap space}.
Unlike trap sets and attractors, trap spaces of a \acbn are independent of the employed update scheme~\citep{KBS2015}.
Then a trap space \(m\) is minimal \ifftext there is no other trap space \(m'\) such that \(\cset{m'} \subset \cset{m}\).
It is easy to derive that a minimal trap space contains at least one attractor of the \acbn regardless of the employed update scheme~\citep{KBS2015}.

\begin{example}\label{exam:throughout-BN}
	Consider \acbn \(f\) with \(\var{f} = \{p, q, r\}, f_p = \neg q, f_q = \neg p, f_r = q\).
	\Cref{fig:BN-IG-SSTG-ASTG}~(a), \Cref{fig:BN-IG-SSTG-ASTG}~(b), and \Cref{fig:BN-IG-SSTG-ASTG}~(c) show the influence graph, the synchronous \acstg, and the asynchronous \acstg of \(f\), respectively.
	Attractor states are highlighted with boxes.
	The synchronous \acstg \(\sstg{f}\) has two fixed points and one cyclic attractor, whereas the asynchronous \acstg \(\astg{f}\) has only two fixed points.
	The \acbn \(f\) has five trap spaces: \(m_1 = 10\star\), \(m_2 = 01\star\), \(m_3 = \star\star\star\), \(m_4 = 100\),  and \(m_5 = 011\).
	Among them, \(m_4\) and \(m_5\) are minimal; they are also two fixed points of \(f\).
\end{example}

\begin{figure}[!ht]
	\centering
	\begin{subfigure}{0.2\textwidth}
		\centering
		\begin{tikzpicture}[node distance=1cm and 1cm, every node/.style={scale=0.8}]
			\node[circle, draw] (p) [] {$p$};
			\node[circle, draw] (q) [below=of p] {$q$};
			\node[circle, draw] (r) [below=of q] {$r$};
			
			\draw[->] (q) edge [bend right=30] node [midway, above, fill=white] {$\ominus$} (p);
			\draw[->] (p) edge [bend right=30] node [midway, above, fill=white] {$\ominus$} (q);
			\draw[->] (q) edge [] node [midway, above, fill=white] {$\oplus$} (r);
		\end{tikzpicture}
		\caption{}
	\end{subfigure}
	\begin{subfigure}{0.35\textwidth}  
		\centering
		\begin{tikzpicture}[node distance=1cm and 1cm, every node/.style={scale=0.8}]
			\node[] (pr) [] {101};
			\node[draw] (p) [below=of pr] {100};
			\node[] (q) [right=of pr] {010};
			\node[draw] (qr) [below=of q] {011};
			\node[] (e) [below=of p] {000};
			\node[] (pqr) [below=of qr] {111};
			\node[draw] (r) [right=of pqr] {001};
			\node[draw] (pq) [right=of qr] {110};
			
			\draw[->] (pr) edge [] (p);
			\draw[->] (q) edge [] (qr);
			\draw[->] (e) edge [] (pq);
			\draw[->] (pqr) edge [] (r);
			\draw[->] (pq) edge [bend left=15] (r);
			\draw[->] (r) edge [bend left=15] (pq);
			
			\draw[->] (p) edge [in=0, out=45, loop] (p);
			\draw[->] (qr) edge [in=0, out=45, loop] (qr);
		\end{tikzpicture}
		\caption{}
	\end{subfigure}
	\begin{subfigure}{0.4\textwidth}  
		\centering
		\begin{tikzpicture}[node distance=1cm and 1cm, every node/.style={scale=0.8}]
			\node[] (pr) [] {101};
			\node[draw] (p) [below=of pr] {100};
			\node[] (q) [right=of pr] {010};
			\node[draw] (qr) [below=of q] {011};
			\node[] (e) [below=of p] {000};
			\node[] (pqr) [below=of qr] {111};
			\node[] (r) [right=of pqr] {001};
			\node[] (pq) [right=of qr] {110};
			
			\draw[->] (pr) edge [] (p);
			\draw[->] (q) edge [] (qr);
			
			\draw[->] (e) edge [] (p);
			\draw[->] (e) edge [] (q);
			
			\draw[->] (pqr) edge [] (qr);
			\draw[->] (pqr) edge [bend right=15] (pr);
			
			\draw[->] (pq) edge [bend left=15] (q);
			\draw[->] (pq) edge [bend left=25] (p);
			\draw[->] (pq) edge [bend left=15] (pqr);
			
			\draw[->] (r) edge [bend left=25] (pr);
			\draw[->] (r) edge [bend left=15] (qr);
			\draw[->] (r) edge [bend left=15] (e);
			
			\draw[->] (p) edge [in=180, out=135, loop] (p);
			\draw[->] (qr) edge [in=0, out=45, loop] (qr);
			
			\draw[->] (pr) edge [out=225, in=180, loop] (pr);
			\draw[->] (q) edge [out=-45, in=0, loop] (q);
		\end{tikzpicture}
		\caption{}
	\end{subfigure}
	\caption{(a) \(\ig{f}\), (b) \(\sstg{f}\), and (c) \(\astg{f}\). The \acbn \(f\) is given in~\Cref{exam:throughout-BN}.}\label{fig:BN-IG-SSTG-ASTG}
\end{figure}
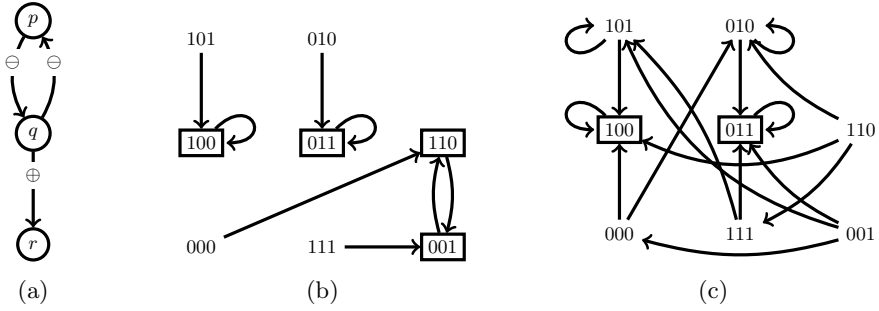

\section{\DLN Programs and Boolean Networks}\label{sec:connection}

To bridge the gap between logic programming and Boolean network analysis, we now introduce a way to encode a \pname as a \acbn. 
This translation enables us to analyze the logical structure and dynamic behavior of a \pname using tools and concepts from the \acbn theory. 
Intuitively, each atom in the \pname becomes a variable in the \acbn, and the update function of each variable reflects the conditions under which the atom can be derived in the program.
This encoding provides a natural and faithful representation of \pnames in terms of Boolean dynamics, and serves as a foundation for structural and dynamical analysis in the subsequent sections.

\begin{definition}\label{def:Datalog-2-BN-encoding}
	Let \(P\) be a \pname.
	We define a BN \(f\) that corresponds to \(P\) as follows: \(\var{f} = \hb{P}\) and for each \(v \in \var{f}\), \[f_v = \bigvee_{r \in \gr{P},\ v = \head{r}}\bodyf{r}.\]
	By convention, if there is no rule \(r \in \gr{P}\) such that \(\head{r} = v\), then \(f_v = \fval\).
\end{definition}

The graphical connection between \pnames and \acbns is given in~\Cref{prop:Datalog-BN-graphs}.
The atom dependency graph is purely syntactic—it reflects all occurrences of atoms in rule bodies, whereas the influence graph is (locally) dynamical—it only includes edges that actually cause a change in the Boolean update function.
Intuitively, the influence graph of the \acbn derived from a \pname is always contained within its atom dependency graph. This reflects the fact that dynamical influence is grounded in syntactic dependency, but not all syntactic dependencies are dynamically active.

This graphical relation is very important because all the subsequent results rely on it and the atom dependency graph of a \pname can be efficiently built based on the syntax only, whereas the construction of the influence graph of a \acbn may be exponential in general.
Note however that the influence graph of a \acbn is usually built by using Binary Decision Diagrams (BDDs)~\citep{Richard2019}.
In this case, the influence graph can be efficiently obtained because each Boolean function is already in Disjunctive Normal Form (DNF) (see~\Cref{def:Datalog-2-BN-encoding}), thus the BDD of this function would be not too large.

\begin{proposition}
  \label{prop:Datalog-BN-graphs}
	Consider a \pname \(P\).
	Let \(f\) be the encoded \acbn of \(P\).
	Then the influence graph of \(f\) is a sub-graph of the atom dependency graph of \(P\).
\end{proposition}
\begin{proof}
	By construction, \(\ig{f}\) and \(\adg{P}\) have the same set of vertices (i.e., \(\hb{P}\)).
	Let \((uv, s)\) be an arc in \(\ig{f}\).
	We show that \((uv, s)\) is also an arc in \(\adg{P}\).
	Without loss of generality, suppose that \(s = \oplus\).
	
	Assume that \((uv, \oplus)\) is not an arc of \(\adg{P}\).
	There are two cases.
	\textbf{Case 1}: there is no arc from \(u\) to \(v\) in \(\adg{P}\).
	In this case, both \(u\) and \(\neg u\) clearly do not appear in \(f_v\).
	This implies that \(\ig{f}\) has no arc from \(u\) to \(v\), which is a contradiction.
	\textbf{Case 2}: there is only a negative arc from \(u\) to \(v\) in \(\adg{P}\).
	It follows that \(\neg u\) appears in \(f_v\) but \(u\) does not because \(f_v\) is in DNF.
	Then, for any state \(x\) and for every conjunction \(c\) of \(f_v\), we have that \(c(x[u \leftarrow 0]) \geq c(x[u \leftarrow 1])\).
	This implies that \(f_v(x[u \leftarrow 0]) \geq f_v(x[u \leftarrow 1])\) for any state \(x\).
	Since \((uv, \oplus)\) is an arc in \(\ig{f}\), there is a state \(x\) such that \(f_v(x[u \leftarrow 0]) < f_v(x[u \leftarrow 1])\).
	This leads to a contradiction.
	Hence, \((uv, \oplus)\) is an arc in \(\adg{P}\).
	
	Now, we can conclude that \(\ig{f}\) is a sub-graph of \(\adg{P}\)).
\end{proof}

The above result establishes a structural correspondence between a \pname and its encoded \acbn. 
This connection lays a foundation for transferring concepts and techniques between the two domains. 
In particular, it motivates the application of ideas from argumentation theory and \acbn dynamics to \pnames. 
Inspired by the concept of \emph{complete extension} in abstract argumentation frameworks, the notion of \emph{complete trap space} has been proposed for \acbns~\citep{TBR2025}.

\begin{definition}[\cite{TBR2025}]\label{def:BN-CoTS}
	Consider a \acbn \(f\).
	A sub-space \(m\) is a \emph{complete trap space} of \(f\) \ifftext for every \(v \in \var{f}\), \(m(v) = m(f_v)\).
\end{definition}

\begin{theorem}
  \label{theo:Datalog-SuPM-BN-CoTS}
	Let \(P\) be a \pname and \(f\) be its encoded \acbn.
	Then the supported partial models of \(P\) coincide with the complete trap spaces of \(f\).
\end{theorem}
\begin{proof}
	A three-valued interpretation \(I\) of \(P\) is a supported partial model of \(P\) \\
	\ifftext \(I\) is a three-valued model of \(\comp{P}\) \\
	\ifftext \(I\) is a three-valued model of \(\bigwedge_{v \in \hb{P}}(v \leftrightarrow f_v)\) \\
	\ifftext \(I\) is a complete trap space of \(f\).
\end{proof}

Building on this result, we now recall a key result concerning the relationship between trap spaces and complete trap spaces in \acbns (see~\Cref{lem:BN-TS-inclusion-CoTS}), thereby proving a direct consequence (see~\Cref{theo:BN-min-CoTS-min-TS}).

\begin{lemma}[Proposition 3 of~\cite{TBR2025}]\label{lem:BN-TS-inclusion-CoTS}
	Let \(f\) be a \acbn.
	For every trap space \(m\) of \(f\), there is a complete trap space \(\widehat{m}\) of \(f\) such that \(\widehat{m} \leq_s m\).
\end{lemma}

\begin{theorem}\label{theo:BN-min-CoTS-min-TS}
	Let \(f\) be a \acbn.
	A sub-space \(m\) is a \(\leq_s\)-minimal complete trap space of \(f\) \ifftext \(m\) is a \(\leq_s\)-minimal trap space of \(f\).
\end{theorem}
\begin{proof}
	Regarding the forward direction, assume that \(m\) is not a \(\leq_s\)-minimal trap space of \(f\).
	Then there is a trap space \(m'\) of \(f\) such that \(m' <_s m\).
	By Lemma~\ref{lem:BN-TS-inclusion-CoTS}, there is a complete trap space \(\widehat{m}\) such that \(\widehat{m} \leq_s m'\).
	It follows that \(\widehat{m} <_s m\), which is a contradiction.
	Hence, \(m\) is a \(\leq_s\)-minimal trap space of \(f\).
	The backward direction is trivial since the set of complete trap spaces is a subset of the set of trap spaces.
\end{proof}

\Cref{theo:BN-min-CoTS-min-TS} shows that minimality under the subset ordering is preserved when restricting to complete trap spaces, thereby aligning the minimal elements in both the sets of trap spaces and complete trap spaces. 
This entails the key connection expressed by
\begin{corollary}
  \label{cor:Datalog-min-SuPM-BN-min-TS}
	Let \(P\) be a \pname and \(f\) be its encoded \acbn.
	Then the \(\leq_s\)-minimal supported partial models of \(P\) coincide with the \(\leq_s\)-minimal trap spaces of \(f\).
\end{corollary}
\begin{proof}
	This immediately follows from~\Cref{theo:BN-min-CoTS-min-TS} and~\Cref{theo:Datalog-SuPM-BN-CoTS}.
\end{proof}

By restricting our attention to two-valued interpretations, we naturally obtain the following correspondence between supported models in \pnames and fixed points in \acbns.

\begin{corollary}\label{cor:Datalog-SuM-BN-fix}
	Let \(P\) be a \pname and \(f\) be its encoded \acbn.
	Then the supported models of \(P\) coincide with the fixed points of \(f\).
\end{corollary}
\begin{proof}
	This immediately follows from~\Cref{cor:Datalog-min-SuPM-BN-min-TS} and the fact that the supported models of \(P\) are the two-valued supported partial models of \(P\) (thus are always \(\leq_s\)-minimal) and the fixed points of \(f\) are the two-valued \(\leq_s\)-minimal trap spaces of \(f\).
\end{proof}

To illustrate the correspondence between \pnames and \acbns established above, we present the following concrete example.

\begin{example}\label{exam:Datalog-2-BN-connection}
	Consider \pname \(P = \{a \leftarrow b; a \leftarrow \dng{b}; b \leftarrow \dng{b}, c; c \leftarrow b\}\).
	We use ``;'' to separate program rules.
	The encoded \acbn \(f\) of \(P\) is: \(\var{f} = \{a, b, c\}, f_a = b \lor \neg b, f_b = \neg b \land c, f_c = b\).
	\Cref{fig:exam-Datalog-2-BN-connection}~(a) and~\Cref{fig:exam-Datalog-2-BN-connection}~(b) show the atom dependency graph of \(P\) and the influence graph of \(f\), respectively.
	We can see that \(V(\ig{f}) = V(\adg{P})\) and \(E(\ig{f}) \subset E(\adg{P})\).
	The \acbn \(f\) has four trap spaces: \(m_1 = \{a = \star, b = 0, c = 0\}\), \(m_2 = \{a = \star, b = \star, c = \star\}\), \(m_3 = \{a = 1, b = 0, c = 0\}\), \(m_4 = \{a = 1, b = \star, c = \star\}\).
	Among others, \(m_2\) and \(m_3\) are two complete trap spaces of \(f\) and are also two supported partial models of \(P\).
	In particular, \(m_3\) is a fixed point of \(f\) and is also a supported model of \(P\).
\end{example}

\begin{figure}[!ht]
	\centering
	\begin{subfigure}[b]{0.5\textwidth}
		\centering
		\begin{tikzpicture}[node distance=1.5cm and 1.5cm]
			\node[circle, draw] (a) [] {$a$};
			\node[circle, draw] (b) [right=of a, xshift=0cm] {$b$};
			\node[circle, draw] (c) [right=of b, xshift=0cm] {$c$};
			
			\draw[->] (b) edge [bend right=25] node [midway, above, fill=white] {$\oplus$} (a);
			\draw[->] (b) edge [bend left=25] node [midway, above, fill=white] {$\ominus$} (a);
			
			\draw[->] (c) edge [bend left=25] node [midway, above, fill=white] {$\oplus$} (b);
			\draw[->] (b) edge [bend left=25] node [midway, above, fill=white] {$\oplus$} (c);
			
			\path [] (b) edge [out=60, in=90, loop] node [midway, above, fill=white] {$\ominus$} (b);
		\end{tikzpicture}
		\caption{}
	\end{subfigure}%
	\begin{subfigure}[b]{0.5\textwidth}
		\centering
		\begin{tikzpicture}[node distance=1.5cm and 1.5cm]
			\node[circle, draw] (a) [] {$a$};
			\node[circle, draw] (b) [right=of a, xshift=0cm] {$b$};
			\node[circle, draw] (c) [right=of b, xshift=0cm] {$c$};
			
			\draw[->] (c) edge [bend left=25] node [midway, above, fill=white] {$\oplus$} (b);
			\draw[->] (b) edge [bend left=25] node [midway, above, fill=white] {$\oplus$} (c);
			
			\path [] (b) edge [out=60, in=90, loop] node [midway, above, fill=white] {$\ominus$} (b);
		\end{tikzpicture}
		\caption{}
	\end{subfigure}%
	\caption{(a) Atom dependency graph \(\adg{P}\) of the \pname \(P\), (b) influence graph \(\ig{f}\) of the \acbn \(f\). The details of \(P\) and \(f\) are given in~\Cref{exam:Datalog-2-BN-connection}.}
	\label{fig:exam-Datalog-2-BN-connection}
\end{figure}
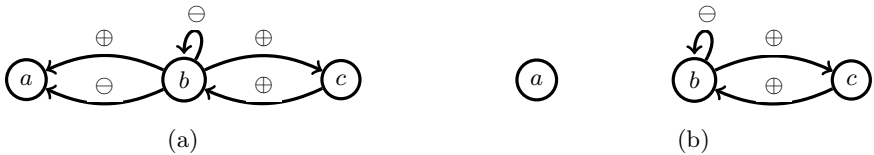

\section{Graphical Analyses of \DLN Programs}\label{sec:graphical-analysis}

In this section, we present new graphical analysis results on \pnames.
We first present some problems in the claims published in~\cite{YY1994} (see~\Cref{subsec:graphical-analysis-revisit}).
We then introduce precise definitions for proving our main results on odd cycles (see~\Cref{subsec:model-existence-odd-cycle}), even cycles (see~\Cref{subsec:model-unicity-even-cycle}), and upper bounds (see~\Cref{subsec:num-model-fvs}).
To improve readability and help the reader quickly identify key findings,~\Cref{tab:summary-graphical-analysis} provides a structured overview of the results we prove for \pnames, with the key contributions highlighted in bold.

\begin{table}[!ht]
	\centering
	\caption{Summary of graphical analysis results on \pnames. Main results are highlighted in bold.}
	\label{tab:summary-graphical-analysis}
	\renewcommand{\arraystretch}{1.3}
	\begin{tabular}{l|p{3.5cm}|p{3.5cm}|p{3.5cm}}
		\toprule
		\textbf{Class} & \textbf{Odd Cycles} & \textbf{Even Cycles} & \textbf{Upper Bounds} \\
		\midrule
		\textbf{General} & \begin{minipage}[t]{\linewidth}
			\begin{itemize}[leftmargin=*]
				\item \textbf{\Cref{theo:Datalog-no-odd-cycle-RegM-StM}}
				\item \Cref{cor:Datalog-no-odd-cycle-StM}
				\item \Cref{theo:Datalog-one-SCC-two-stable-models}
			\end{itemize}
		\end{minipage} & \begin{minipage}[t]{\linewidth}
		\begin{itemize}[leftmargin=*]
			\item \textbf{\Cref{theo:Datalog-no-even-cycle-unique-SuPM}}
			\item \Cref{cor:Datalog-no-even-cycle-unique-StPM}
			\item \Cref{cor:Datalog-no-even-cycle-unique-RegM}
			\item \Cref{cor:Datalog-no-even-cycle-at-most-one-StM}
		\end{itemize}
		\end{minipage} & \begin{minipage}[t]{\linewidth}
		\begin{itemize}[leftmargin=*]
			\item \Cref{theo:Datalog-num-SuPM-even-cycle}
			\item \Cref{cor:Datalog-num-StPM-RegM-StM-even-cycle}
			\item \Cref{prop:Datalog-num-RegM-non-tight-trivial}
			\item \Cref{theo:Datalog-num-RegM-tight-even-cycle}
			\item \Cref{theo:Datalog-num-StM-general-even-cycle}
		\end{itemize}
		\end{minipage} \\
		\midrule
		\textbf{Uni-rule} & \begin{minipage}[t]{\linewidth}
			\begin{itemize}[leftmargin=*]
				\item \Cref{cor:uni-rule-Datalog-tight-odd-cycle-internal-del-triple-StM}
				\item \textbf{\Cref{theo:uni-rule-Datalog-tight-odd-cycle-internal-del-triple-RegM}}
			\end{itemize}
		\end{minipage} & \begin{minipage}[t]{\linewidth}
		\begin{itemize}[leftmargin=*]
			\item \Cref{theo:uni-rule-Datalog-no-even-cycle-del-triple-at-most-one-StM}
			\item \Cref{theo:tight-uni-rule-Datalog-no-even-cycle-del-triple-unique-RegM}
		\end{itemize}
		\end{minipage} & \begin{minipage}[t]{\linewidth}
		\begin{itemize}[leftmargin=*]
			\item \Cref{theo:uni-rule-Datalog-num-StM-even-cycle-del-triple}
			\item \Cref{theo:uni-rule-Datalog-num-RegM-tight-even-cycle-del-triple}
		\end{itemize}
		\end{minipage} \\
		\bottomrule
	\end{tabular}
\end{table}

\subsection{Previous Claims Made for Normal Logic Programs}\label{subsec:graphical-analysis-revisit}

In~\cite{YY1994}, the authors used the following notion of stratification which exists for any \nlp:

\begin{definition}[\cite{YY1994}]\label{def:NLP-stratification}
	Let \(P\) be an \nlp.
	A \emph{stratification} of \(P\) is a partial order \(\leq\) over subsets of \(\hb{P}\) such that
	\begin{itemize}
		\item every (ground) atom belongs to one and only one stratum; and
		\item two atoms \(a\) and \(b\) are in the same stratum if they are on a common cycle in \(\adg{P}\), or there exists an atom \(c\) such that \(a\) and \(c\) are in the same stratum and the same holds true for \(b\) and \(c\); and these are the only atoms that can be in the same stratum.
	\end{itemize}
	Let \([a]\) denote the stratum of an atom \(a\).
	\([a]\) is a lower stratum than \([b]\), denoted by \([a] \leq [b]\), \ifftext there is a path from some atom in \([a]\) to some atom in \([b]\).
\end{definition}

\begin{definition}[\cite{YY1994}]\label{def:NLP-wf-stratification-wrong}
	Let \(P\) be an \nlp.
	A stratification of \(P\) is said to be \emph{well-founded} \ifftext for every stratum \([b]\), there exists \([a]\) such that \([a] \leq [b]\) and for any stratum \([c]\), if \([c] \leq [a]\) then there are only positive arcs from atoms in \([c]\) to atoms in \([a]\).
\end{definition}

However, the given \Cref{def:NLP-wf-stratification-wrong} appears to be inconsistent with the proofs presented in~\cite{YY1994}, since they do not use the positive arc requirement, as well as with their later definition given in~\cite{LY2002}.
Let us thus consider the following standard definition:
\begin{definition}[well-founded stratification]\label{def:NLP-wf-stratification-right}
	Let \(P\) be an NLP.
	A stratification of \(P\) is said to be \emph{well-founded} \ifftext there is no infinite descending chain of strata \([a_0] > [a_1] > [a_2] > \dots\).
\end{definition}

With that definition, an \nlp whose Herbrand base is finite always has a well-founded stratification.
In contrast, it is well known that a general \nlp may not have a well-founded stratification, e.g.,~\nlp \(\{p(X) \leftarrow p(s(X))\}\).
By mimicking the proof of Theorem 5.3(i) of~\cite{YY1994}, we can prove the following result for negative \nlps:

\begin{theorem}\label{theo:You-Yuan-Theorem-5-3-i}
	Consider a \textbf{negative} \nlp \(P\).
	If \(P\) has a well-founded stratification and \(\adg{P}\) has no odd cycle, then all the regular models of \(P\) are two-valued.
\end{theorem}

The result for general \nlps claimed in~\cite{YY1994} is unfortunately not correctly proved there, nor in any subsequent work to the best of our knowledge~\citep{You25}.
We can similarly prove Theorem 5.3(ii) of~\cite{YY1994} for negative \nlps only:
\begin{theorem}\label{theo:You-Yuan-Theorem-5-3-ii}
	Consider a \textbf{negative} \nlp \(P\).
	If \(P\) has a well-founded stratification and \(\adg{P}\) has no even cycle, then \(P\) has a unique regular model.
\end{theorem}

That second result is claimed in~\cite{YY1994} for general \nlps using the property that ``if \(\adg{P}\) has no even cycle, then \(\adg{\lfp{P}}\) has no even cycle'', but this is not correct:

\begin{counter-example}
Let \(P = \{a \leftarrow c; b \leftarrow c; c \leftarrow \dng{a}, \dng{b}\}\).
Then \(\lfp{P} = \{a \leftarrow \dng{a}, \dng{b}; b \leftarrow \dng{a}, \dng{b}; c \leftarrow \dng{a}, \dng{b}\}\).
The graph \(\adg{P}\) has no even cycle, but \(\adg{\lfp{P}}\) does (there is the even cycle \(a \xrightarrow{\ominus} b \xrightarrow{\ominus} a\) in \(\adg{\lfp{P}}\)).
\end{counter-example}

Below, we fix these issues for general \pnames with~\Cref{theo:Datalog-no-odd-cycle-RegM-StM} and~\Cref{cor:Datalog-no-even-cycle-unique-RegM},
but it is worth noting that the questions for general \nlps are still open to the best of our knowledge~\citep{You25}.

\subsection{Definitions for Graphical Analysis}\label{subsec:graphical-analysis-preparations}

To facilitate our graphical analysis of \pnames, we begin by establishing a set of technical definitions and auxiliary results. 
These foundational elements provide the necessary tools to reason about the structural properties of programs via their associated graphs. 
In particular, we focus on the interplay between atom dependencies, program rules, and semantic interpretations, setting the stage for the results that follow.

We first start with the natural yet important insight that the least fixpoint of a \pname is also a \pname.

\begin{proposition}\label{prop:LFP-finiteness}
	Consider a \pname \(P\).
	Then the least fixpoint of \(P\) is also a \pname.
\end{proposition}
\begin{proof}
	Let \(\lfp{P}\) be the least fixpoint of \(P\).
	By the definition of least fixpoint, \(\hb{\lfp{P}} = \hb{P}\).
	Hence, \(\lfp{P}\) is also a \pname.
\end{proof}

We now recall and establish several important results that characterize the relationships between different semantics of interpretations under the assumption of \emph{tightness}—a key structural condition on the atom dependency graph of a \pname.
These results highlight the equivalence between stable and supported (partial) models in tight \pnames, and they play a central role in connecting logic programming semantics with the dynamics of associated \acbns. 
In particular, tightness guarantees that the distinction between the stable and supported semantics collapses, thereby simplifying semantic analysis. 
Moreover, the regular models of a tight program coincide with the \(\leq_s\)-minimal trap spaces of the corresponding \acbn, reinforcing the utility of tightness in both logical and graphical reasoning.

\begin{theorem}[Theorem 3.2 of~\cite{Fages1994}]\label{theo:tight-Datalog-StM-SuM}
	Consider a tight \pname \(P\).
	Then the set of stable models of \(P\) coincides with the set of supported models of \(P\).
\end{theorem}

\begin{theorem}[Lemma 16 of~\cite{DHW2014}]\label{theo:tight-Datalog-StPM-SuPM}
	Consider a tight \pname \(P\).
	Then the set of stable partial models of \(P\) coincides with the set of supported partial models of \(P\).
\end{theorem}

\begin{corollary}\label{cor:neg-Datalog-StPM-SuPM}
	Consider a negative \pname \(P\).
	Then the set of stable partial models of \(P\) coincides with the set of supported partial models of \(P\).
\end{corollary}
\begin{proof}
	Since \(P\) is negative, the atom dependency graph of \(P\) contains no positive arcs, thus \(P\) is tight.
	The corollary immediately follows from~\Cref{theo:tight-Datalog-StPM-SuPM}.
\end{proof}

\begin{lemma}\label{lem:Datalog-tight-BN-RegM-s-min-TS}
	Consider a \pname \(P\) and its encoded \acbn \(f\).
	If \(P\) is tight, then the regular models of \(P\) coincide with the \(\leq_s\)-minimal trap spaces of \(f\).
\end{lemma}
\begin{proof}
	A three-valued interpreration \(I\) is a regular model of \(P\) \\
	\ifftext \(I\) is a \(\leq_s\)-minimal stable partial model of \(P\) by definition \\
	\ifftext \(I\) is a \(\leq_s\)-minimal supported partial model of \(P\) by~\Cref{theo:tight-Datalog-StPM-SuPM} \\
	\ifftext \(I\) is a \(\leq_s\)-minimal trap space of \(f\) by~\Cref{cor:Datalog-min-SuPM-BN-min-TS}.
\end{proof}

The following result shows that the (stable) semantic behaviors of a \pname are preserved under its least fixpoint transformation.
However, the (supported) semantic behaviors of a \pname may be not preserved.

\begin{theorem}[Theorem 3.1 of~\cite{AD1995}]\label{theo:Datalog-lfp-model-equivalence}
	Consider a \pname \(P\).
	Let \(\lfp{P}\) be the least fixpoint of \(P\).
	Then \(P\) and \(\lfp{P}\) have the same set of stable partial models, the same set of regular models, and the same set of stable models.
\end{theorem}

The next result establishes a characterization of regular models via trap spaces after applying the least fixpoint transformation to a \pname.

\begin{lemma}\label{lem:Datalog-lfp-BN-RegM-s-min-TS}
	Consider a \pname \(P\).
	Let \(\lfp{P}\) be the least fixpoint of \(P\) and \(f'\) be the encoded \acbn of \(\lfp{P}\).
	Then the regular models of \(P\) coincide with the \(\leq_s\)-minimal trap spaces of \(f'\).
\end{lemma}
\begin{proof}
	A three-valued interpreration \(I\) is a regular model of \(P\) \\
	\ifftext \(I\) is a regular model of \(\lfp{P}\) by~\Cref{theo:Datalog-lfp-model-equivalence} \\
	\ifftext \(I\) is a \(\leq_s\)-minimal stable partial model of \(\lfp{P}\) by definition \\
	\ifftext \(I\) is a \(\leq_s\)-minimal supported partial model of \(\lfp{P}\) by~\Cref{cor:neg-Datalog-StPM-SuPM} \\
	\ifftext \(I\) is a \(\leq_s\)-minimal trap space of \(f'\) by~\Cref{cor:Datalog-min-SuPM-BN-min-TS}.
\end{proof}

The following example illustrates the concepts and results introduced above by analyzing a concrete \pname, its least fixpoint, and the correspondence between regular models and trap spaces.

\begin{example}\label{exam:Datalog-preparations}
	Consider \pname \(P\) of~\Cref{exam:Datalog-2-BN-connection}.
	 \Cref{fig:exam-Datalog-2-BN-connection}~(b) shows the atom dependency graph of \(P\), which demonstrates that \(P\) is non-tight because of the cycle \(b \xrightarrow{\oplus} c \xrightarrow{\oplus} b\).
	Program \(P\) has two supported partial models: \(m_2 = \{a = \star, b = \star, c = \star\}\) and \(m_3 = \{a = 1, b = 0, c = 0\}\).
	However, only \(m_3\) is a stable partial model of \(P\).
	The least fixpoint of \(P\) is \(\lfp{P} = \{a \leftarrow \dng{b}\}\).
	The program \(\lfp{P}\) has a unique stable (supported) partial model \(m_3\).
	The encoded \acbn \(f'\) of \(\lfp{P}\) is \(\var{f'} = \{a, b, c\}, f'_a = \neg b, f'_b = 0, f'_c = 0\).
	The \acbn \(f'\) has a unique \(\leq_s\)-minimal trap space \(m_3\).
	This is consistent with~\Cref{lem:Datalog-lfp-BN-RegM-s-min-TS}.
\end{example}

\subsection{Model Existence and Odd Cycles}\label{subsec:model-existence-odd-cycle}

This subsection investigates the existence of models for \pnames in the absence of odd cycles in their atom dependency graphs.
We present both known and novel results that clarify the role of odd cycles in determining the presence or absence of various semantic models.
Our analysis distinguishes between general \pnames and the more restrictive class of uni-rule \pnames, for which stronger guarantees and tighter characterizations can be obtained. 
These results lay the foundation for understanding the structural limitations imposed by odd dependencies.

\paragraph{General \DLN Programs.}

We first show that the absence of odd cycles in a \pname is preserved under its least fixpoint transformation.

\begin{lemma}\label{lem:Datalog-lfp-no-odd-cycle}
	Let \(P\) be a \pname and \(\lfp{P}\) be its least fixpoint.
	If \(\adg{P}\) is has no odd cycle, then \(\adg{\lfp{P}}\) has no odd cycle.
\end{lemma}
\begin{proof}
	This directly follows from Lemma 5.3 of~\cite{Fages1994}.
\end{proof}

Building on this observation, we now connect the absence of odd cycles in the atom dependency graph to dynamical properties of the associated \acbn, which ultimately allows us to characterize the nature of regular models in such cases.

\begin{theorem}[Theorem 1 of~\cite{R2010}]\label{theo:BN-no-odd-cycle-no-cyclic-asyn-att}
	Let \(f\) be a \acbn\@.
	If \(\ig{f}\) has no odd cycle, then \(\astg{f}\) has no cyclic attractor.
\end{theorem}

\begin{theorem}[\textbf{main result}]\label{theo:Datalog-no-odd-cycle-RegM-StM}
	Let \(P\) be a \pname.
	If \(\adg{P}\) has no odd cycle, then all the regular models of \(P\) are two-valued.
\end{theorem}
\begin{proof}
	Let \(\lfp{P}\) be the least fixpoint of \(P\).
	By~\Cref{prop:LFP-finiteness}, \(\lfp{P}\) is a \pname.
	By~\Cref{lem:Datalog-lfp-no-odd-cycle}, \(\adg{\lfp{P}}\) has no odd cycle.
	Let \(f'\) be the encoded \acbn of \(\lfp{P}\).
	By~\Cref{lem:Datalog-lfp-BN-RegM-s-min-TS}, the regular models of \(P\) coincide with the \(\leq_s\)-minimal trap spaces of \(f'\).
	
	Since \(\ig{f'}\) is a sub-graph of \(\adg{\lfp{P}}\), \(\ig{f'}\) has no odd cycle.
	By~\Cref{theo:BN-no-odd-cycle-no-cyclic-asyn-att}, \(\astg{f'}\) has no cyclic attractor.
	Each \(\leq_s\)-minimal trap space of \(f'\) contains at least one attractor of \(\astg{f'}\)~\citep{KBS2015}.
	In addition, if a \(\leq_s\)-minimal trap space contains a fixed point, then it is also a fixed point because of the minimality.
	Hence, all the \(\leq_s\)-minimal trap spaces of \(f'\) are fixed points.
	This implies that all the regular models of \(P\) are two-valued.
\end{proof}

An immediate consequence of the above result is the guaranteed existence of stable models for programs whose atom dependency graphs are free of odd cycles.

\begin{corollary}\label{cor:Datalog-no-odd-cycle-StM}
	Consider a \pname \(P\).
	If \(\adg{P}\) has no odd cycle, then \(P\) has at least one stable model.
\end{corollary}
\begin{proof}
	This immediately follows from~\Cref{theo:Datalog-no-odd-cycle-RegM-StM} and the fact that \(P\) always has at least one regular model.
\end{proof}

Inspired by~\Cref{cor:Datalog-no-odd-cycle-StM}, we explore an interesting result shown in~\Cref{theo:Datalog-one-SCC-two-stable-models}.
To prove this result, we first prove an auxiliary result that establishes a useful structural property of signed directed graphs that are strongly connected (see~\Cref{lem:sign-definite-graph-SCC}).
Then an existing result in \acbns is applied.

\begin{lemma}\label{lem:sign-definite-graph-SCC}
	If a signed directed graph \(G\) is strongly connected and has no odd cycle or has no even cycle, then \(G\) is sign-definite.
\end{lemma}
\begin{proof}
	We first prove that each arc of \(G\) belongs to a cycle in \(G\) (*).
	Taken an arbitrary arc \((uv, \epsilon)\) in \(G\) where \(\epsilon \in \{\oplus, \ominus\}\).
	Since \(G\) is strongly connected, there is a directed path from \(v\) to \(u\).
	By adding \((uv, \epsilon)\) to this path, we obtain a cycle.
	
	Assume that \(G\) is not sign-definite.
	Then there are two arcs: \((uv, \oplus)\) and \((uv, \ominus)\).
	By (*), \((uv, \oplus)\) (resp. \((uv, \ominus)\)) belongs to a cycle in \(G\) (say \(C\)).
	\(C\) is an even (resp. odd) cycle because \(G\) has no odd (resp. even) cycle.
	Then \((C - (uv, \oplus)) + (uv, \ominus)\) (resp. \((C - (uv, \ominus)) + (uv, \oplus)\)) is an odd (resp. even) cycle in \(G\).
	This implies a contradiction.
	Hence, \(G\) is sign-definite.
\end{proof}

\begin{theorem}\label{theo:Datalog-one-SCC-two-stable-models}
	Consider a \pname \(P\).
	Suppose that \(\adg{P}\) is strongly connected, has at least one arc, and has no odd cycle.
	If \(P\) is tight, then \(P\) has two stable models \(A\) and \(B\) such that \(\forall v \in \hb{P}\), either \(v \in A\) or \(v \in B\).
	In addition, \(A\) and \(B\) can be computed in polynomial time \wrttext \(|\hb{P}|\).
\end{theorem}
\begin{proof}
	Let \(f\) be the encoded \acbn of \(P\).
	By~\Cref{theo:tight-Datalog-StM-SuM}, the stable models of \(P\) coincide with the supported models of \(P\), thus the fixed points of \(f\) by~\Cref{cor:Datalog-SuM-BN-fix}.
	We show that \(f\) has two fixed points that are complementary.
	
	Since \(\adg{P}\) is strongly connected and has no negative cycle, it is sign-definite by~\Cref{lem:sign-definite-graph-SCC}.
	Since \(E(\ig{f}) \subseteq E(\adg{P})\), \(\ig{f}\) is also sign-definite.
	The graph \(\adg{P}\) has the minimum in-degree of at least one because it is strongly connected and has at least one arc.
	This implies that for every variable \(j \in \var{f}\), the number of arcs ending at \(j\) in \(\ig{f}\) is at least one.
	Hence, \(f_j\) cannot be constant for every variable \(j \in \var{f}\).
	It is known that when \(\adg{P}\) is strongly connected and has no negative cycle, its set of vertices can be divided into two equivalence classes (say \(S^+\) and \(S^-\)) such that any two vertices in \(S^+\) (resp. \(S^-\)) are connected by either no arc or a positive arc, and there is either no arc or a negative arc between two vertices in \(S^+\) and \(S^-\) (Theorem 1 of~\cite{AMT12}).
	Since \(E(\ig{f}) \subseteq E(\adg{P})\) and \(V(\ig{f}) = V(\adg{P})\), \(S^+\) and \(S^-\) are still such two equivalence classes in \(\ig{f}\).
	
	Let \(x\) be a state defined as: \(x_i = 1\) if \(i \in S^+\) and \(x_i = 0\) if \(i \in S^-\).
	Consider a variable \(j\).
	If \(x_j = 0\), then \(j \in S^{-}\), and for all \(i \in \var{f}\) such that \(\ig{f}\) has a positive arc from \(i\) to \(j\), \(i \in S^{-}\), thus \(x_i = 0\), and for all \(i \in \var{f}\) such that \(\ig{f}\) has a negative arc from \(i\) to \(j\), \(i \in S^{+}\), thus \(x_i = 1\).
	Since \(f_j\) cannot be constant, \(f_j(x) = 0\).
	Analogously, if \(x_j = 1\), then \(f_j(x) = 1\).
	Since \(j\) is arbitrary, we can conclude that \(x\) is a fixed point of \(f\).
	By using the similar deduction, we can conclude that \(\overline{x}\) is also a fixed point of \(f\) where \(\overline{x}_i = 1 - x_i, \forall i \in \var{f}\).
	Let \(A\) and \(B\) are two two-valued interpretations of \(P\) corresponding to \(x\) and \(\overline{x}\).
	Then \(A\) and \(B\) are stable models of \(P\).
	We have that \(\forall v \in \hb{P}\), either \(v \in A\) or \(v \in B\).
	In addition, since \(S^+\) and \(S^-\) can be computed in polynomial time~\citep{AMT12} \wrttext \(|V(\ig{f})| = |\var{f}|\), \(A\) and \(B\) can be computed in polynomial time \wrttext \(|\hb{P}|\).
\end{proof}

The following example demonstrates an application of~\Cref{theo:Datalog-one-SCC-two-stable-models}, showcasing a \pname whose atom dependency graph satisfies the required conditions and admits exactly two complementary stable models.

\begin{example}\label{exam:Datalog-one-SCC-two-stable-models}
	Consider a \pname \(P = \{a \leftarrow \dng{b}; b \leftarrow \dng{a}; b \leftarrow \dng{c}; c \leftarrow \dng{b}\}\).
	\Cref{fig:Datalog-one-SCC-two-stable-models} shows the atom dependency graph of \(P\).
	It is to see that \(\adg{P}\) is strongly connected, has at least one arc, has no odd cycle, and \(P\) is tight.
	The program \(P\) has two supported models: \(A = \{a = \fval, b = \tval, c = \fval\}, B = \{a = \tval, b = \fval, c = \tval\}\).
	These models are also two stable models of \(P\).
	It is easy to see that \(A \cap B = \emptyset\) and \(A \cup B = \hb{P}\).
\end{example}

\begin{figure}[!ht]
	\centering
	\begin{tikzpicture}[node distance=1.5cm and 1.5cm]
		\node[circle, draw] (a) [] {$a$};
		\node[circle, draw] (b) [right=of a, xshift=0cm] {$b$};
		\node[circle, draw] (c) [right=of b, xshift=0cm] {$c$};
		
		\draw[->] (a) edge [bend left=25] node [midway, above, fill=white] {$\ominus$} (b);
		\draw[->] (b) edge [bend left=25] node [midway, above, fill=white] {$\ominus$} (a);
		
		\draw[->] (c) edge [bend left=25] node [midway, above, fill=white] {$\ominus$} (b);
		\draw[->] (b) edge [bend left=25] node [midway, above, fill=white] {$\ominus$} (c);
	\end{tikzpicture}
	\caption{Atom dependency graph of the \pname \(P\) given in~\Cref{exam:Datalog-one-SCC-two-stable-models}.}
	\label{fig:Datalog-one-SCC-two-stable-models}
\end{figure}
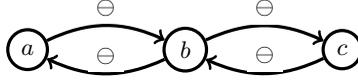

\paragraph{Uni-rule \DLN Programs.}

We now turn our attention to a syntactic fragment of \pnames, which we refer to as \emph{uni-rule} \pnames. 
These are programs in which each ground atom appears in the head of at most one ground rule.
Despite their restricted form, uni-rule \pnames arise naturally in various modeling scenarios and exhibit desirable computational properties (see more detailed discussions at~\Cref{re:uni-rule-Datalog-motivations}).
In this subsection, we study the model existence and semantic characteristics of such programs in relation to the structure of their atom dependency graphs.

\begin{remark}\label{re:uni-rule-Datalog-motivations}
	Why a \upname is important?
	First, the class of \upnames is easily identifiable by a syntatic characterization.
	Second, this class is as hard as the class of general \pnames w.r.t the stable model semantics, since identifying whether a \upname has a stable model or not is NP-complete~\citep{SS1997}.
	In addition, computing the (three-valued) well-founded model of a \upname is linear~\citep{SS1997}.
	Third, the corresponding \pname of an abstract argumentation framework is exactly a \upname~\citep{CG2009,CSAD2015}.
	Finally, it has been shown that in general no abstract argumentation semantics is able to coincide with the L-stable model semantics in \pnames, thus the crucial question is whether there exists a restricted class of \pnames for which the semi-stable semantics in abstract argumentation frameworks does coincide with the L-stable model semantics in \pnames~\citep{CSAD2015}.
	In~\cite{CSAD2015}, the L-stable model semantics of a \upname\footnote{In that paper, a \upname is called an \emph{AF-program}.} coincides with the semi-stable semantics of its corresponding abstract argumentation framework.
\end{remark}

By considering \upnames, we obtain several stronger results.
First, we show that the class of \acbns encoding \upnames is \anbns.

\begin{proposition}\label{prop:uni-rule-Datalog-2-AND-NOT-BN}
	Consider a \upname \(P\).
	Then the encoded \acbn \(f\) of \(P\) is an \anbn.
\end{proposition}
\begin{proof}
	Consider a variable \(v \in \hb{P}\).
	There are two cases.
	\textbf{Case 1}: there is no rule whose head is \(v\).
	Then \(f_v = 0\) by construction.
	\textbf{Case 2}: There is exactly one rule \(r\) whose head is \(v\).
	Then \(f_v = \bigwedge_{w \in \pbody{r}}w \land \bigwedge_{w \in \nbody{r}}\neg w\) by construction.
	It follows that \(f_v\) is either 0 or a conjunction of literals.
	Hence, \(f\) is an \anbn.
\end{proof}

A key advantage of uni-rule \pnames lies in the simplicity of their syntactic structure, which enables a tighter correspondence between their logical and dynamical representations.
The following insight formalizes this observation by establishing that, for any \pname, the atom dependency graph coincides exactly with the influence graph of its associated \acbn.
This structural alignment is significant because it allows one to analyze model-theoretic and dynamical properties of such programs interchangeably through either graph, thereby facilitating the transfer of results and intuitions across the logic programming and \acbn domains.

\begin{proposition}\label{prop:uni-rule-Datalog-AND-NOT-BN-graphs}
	Consider a \upname \(P\).
	Let \(f\) be the encoded \acbn of \(P\).
	Then the influence graph of \(f\) and the atom dependency graph of \(P\) coincide.
\end{proposition}
\begin{proof}
	By construction, \(V(\ig{f}) = V(\adg{P}) = \hb{P}\).
	By~\Cref{prop:Datalog-BN-graphs}, \(E(\ig{f}) \subseteq E(\adg{P})\).
	Assume that \((uv, \oplus)\) be an arc in \(E(\adg{P})\).
	There exists a rule \(r \in \gr{P}\) such that \(\head{r} = v\) and \(u \in \pbody{r}\).
	Let \(x\) be a state of \(f\) such that \(x(w) = 0\) if \(w \in \nbody{r}\), \(x(w) = 1\) if \(w \in \pbody{r}\), and \(x(w) = 1\) otherwise.
	Since \(P\) is uni-rule, \(f_v = \bigwedge_{w \in \pbody{r}}w \land \bigwedge_{w \in \nbody{r}}\neg w\).
	We have that \(f_v(x[u = 0]) = 0 < f_v(x[u = 1]) = 1\).
	Hence, \((uv, \oplus)\) is also an arc in \(E(\ig{f})\).
	The case that \((uv, \ominus)\) is an arc in \(E(\adg{P})\) is similar.
	It follows that \(E(\adg{P}) \subseteq E(\ig{f})\), leading to \(E(\ig{f}) = E(\adg{P})\).
	Hence, \(\ig{f}\) coincides with \(\adg{P}\).
\end{proof}

\begin{remark}
	It is easy to see that the ground instantiation of a \upname is uniquely determined by its atom dependency graph.
	Similarly, an \anbn is uniquely determined by its influence graph.
\end{remark}

The notion of \emph{delocalizing triple} plays a central role in analyzing the structural properties of signed directed graphs that underlie both uni-rule \pnames and AND-NOT \acbns. 
As established earlier, these two formalisms are uniquely specified through their respective graphs: the atom dependency graph in the case of \pnames, and the influence graph for \acbns. 
Within this unified graphical perspective, delocalizing triples—introduced by~\citet{RR2013}—serve as critical structural motifs that can disrupt cyclic behaviors and affect the existence or uniqueness of fixed points. The following definition formalizes this concept, followed by a concrete example.

\begin{definition}[\cite{RR2013}]\label{def:signed-digraph-delocalizing-triple}
	Given a signed directed graph \(G\), a cycle \(C\) of \(G\), and vertices \(u\), \(v_1\), \(v_2\) of \(G\), \((u, v_1, v_2)\) is said to be a \emph{\delt} of \(C\) when 1) \(v_1\), \(v_2\) are distinct vertices of \(C\); 2) \((uv_1, \oplus)\) and  \((uv_2, \ominus)\) are arcs of \(G\) that are not in \(E(C)\).
	Such a \delt is called \emph{internal} if \(u \in V(C)\), and \emph{external} otherwise.
\end{definition}

\begin{example}\label{exam:AND-NOT-BN-del-triple}
	Consider the signed directed graph \(G\) taken from~\cite[Figure 4]{VBHKWL12}.
	It is also shown in~\Cref{fig:AND-NOT-BN-del-triple}.
	Regarding cycle \(C_1 = v_3 \xrightarrow{\ominus} v_4 \xrightarrow{\ominus} v_3\) of \(G\), \((v_1, v_3, v_4)\) is an external \delt of \(C_1\) because \((v_1v_3, \oplus)\) and \((v_1v_4, \ominus)\) are arcs of \(G\) but not \(C_1\) and \(v_1 \not \in V(C_1)\).
	Regarding cycle \(C_2 = v_1 \xrightarrow{\oplus} v_2 \xrightarrow{\ominus} v_4 \xrightarrow{\ominus} v_3 \xrightarrow{\oplus} v_5 \xrightarrow{\oplus} v_1\) of \(G\), \((v_1, v_5, v_4)\) is an internal \delt of \(C_2\) because \((v_1v_5, \oplus)\) and \((v_1v_4, \ominus)\) are arcs of \(G\) but not \(C_2\) and \(v_1 \in V(C_2)\).
	Every remaining cycle of \(G\) has no \delt.
\end{example}

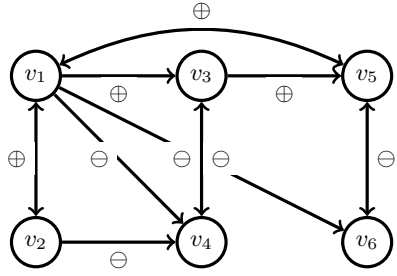
\begin{figure}[!ht]
	\centering
	\begin{tikzpicture}[node distance=1.5cm and 1.5cm]
		\node[circle, draw] (v1) [] {$v_1$};
		\node[circle, draw] (v2) [below=of v1, xshift=0cm] {$v_2$};
		\node[circle, draw] (v3) [right=of v1, xshift=0cm] {$v_3$};
		\node[circle, draw] (v4) [right=of v2, xshift=0cm] {$v_4$};
		\node[circle, draw] (v5) [right=of v3, xshift=0cm] {$v_5$};
		\node[circle, draw] (v6) [right=of v4, xshift=0cm] {$v_6$};
		
		\draw[<->] (v1) edge [] node [midway, left, fill=white] {$\oplus$} (v2);
		\draw[->] (v1) edge [] node [midway, left, fill=white] {$\ominus$} (v4);
		\draw[->] (v1) edge [] node [midway, below, fill=white] {$\oplus$} (v3);
		\draw[->] (v1) edge [] node [midway, left, fill=white] {$\ominus$} (v6);
		
		\draw[->] (v2) edge [] node [midway, below, fill=white] {$\ominus$} (v4);
		
		\draw[<->] (v3) edge [] node [midway, right, fill=white] {$\ominus$} (v4);
		\draw[->] (v3) edge [] node [midway, below, fill=white] {$\oplus$} (v5);
		
		\draw[<->] (v5) edge [bend right=25] node [midway, above, fill=white] {$\oplus$} (v1);
		\draw[<->] (v5) edge [] node [midway, right, fill=white] {$\ominus$} (v6);
	\end{tikzpicture}
	\caption{Signed directed graph \(G\) of~\Cref{exam:AND-NOT-BN-del-triple}.}
	\label{fig:AND-NOT-BN-del-triple}
\end{figure}

Cycles in the influence graph of an AND-NOT \acbn (or equivalently, the atom dependency graph of a uni-rule \pname) can critically impact the existence of fixed points or stable models.
\citet{RR2013} showed that the presence of internal delocalizing triples within odd cycles imposes the stability of fixed points.
The following theorem formalizes this result, guaranteeing the existence of a fixed point for any AND-NOT \acbn whose odd cycles are all ``internally delocalized.''
As a direct consequence, we obtain a sufficient condition for the existence of stable models in tight uni-rule \pnames.

\begin{theorem}[Theorem 3' of~\cite{RR2013}]\label{theo:AND-NOT-BN-fix-odd-cycle-internal-del-triple}
	Let \(f\) be an \anbn.
	If every odd cycle of \(\ig{f}\) has an internal \delt, then \(f\) has at least one fixed point.
\end{theorem}

\begin{corollary}\label{cor:uni-rule-Datalog-tight-odd-cycle-internal-del-triple-StM}
	Consider a \upname \(P\).
	If \(P\) is tight and every odd cycle of \(\adg{P}\) has an internal \delt, then \(P\) has at least one stable model.
\end{corollary}
\begin{proof}
	Let \(f\) be the encoded \acbn of \(P\).
	By~\Cref{cor:Datalog-SuM-BN-fix}, the supported models of \(P\) coincide with the fixed points of \(f\).
	Since \(\ig{f} = \adg{P}\) by~\Cref{prop:uni-rule-Datalog-AND-NOT-BN-graphs}, every odd cycle of \(\ig{f}\) has an internal \delt.
	By~\Cref{theo:AND-NOT-BN-fix-odd-cycle-internal-del-triple}, \(f\) has at least one fixed point.
	Since \(P\) is tight, the supported models of \(P\) coincide with the stable models of \(P\) by~\Cref{theo:tight-Datalog-StM-SuM}.
	Hence, \(P\) has at least one stable model.
\end{proof}

We generalize the above result for stable models to regular models (see~\Cref{theo:uni-rule-Datalog-tight-odd-cycle-internal-del-triple-RegM}).

\begin{theorem}[\textbf{main result}]\label{theo:uni-rule-Datalog-tight-odd-cycle-internal-del-triple-RegM}
	Consider a \upname \(P\).
	If \(P\) is tight and every odd cycle of \(\adg{P}\) has an internal \delt, then every regular model of \(P\) is two-valued.
\end{theorem}
\begin{proof}
	Let \(f\) be the encoded \acbn of \(P\).
	By~\Cref{prop:uni-rule-Datalog-2-AND-NOT-BN}, \(f\) is an \anbn.
	By~\Cref{prop:uni-rule-Datalog-AND-NOT-BN-graphs}, \(\ig{f} = \adg{P}\).
	It follows that every odd cycle of \(\ig{f}\) has an internal \delt.
	Since \(P\) is tight, by~\Cref{lem:Datalog-tight-BN-RegM-s-min-TS}, the regular models of \(P\) coincide with the \(\leq_s\)-minimal trap spaces of \(f\).
	
	Assume that \(m\) is a non-trivial \(\leq_s\)-minimal trap space of \(f\) (i.e., \(m\) is not two-valued).
	We build the new \acbn \(f'\) as follows: for every \(v \in \var{f}\) and \(m(v) \neq \star\), \(f'_v = m(v)\); and for every \(v \in \var{f}\) and \(m(v) = \star\), \(f'_v = f_v\).
	
	It is easy to derive that \(\ig{f'}\) is a sub-graph of \(\ig{f}\).
	Hence, an odd cycle in \(\ig{f'}\) is also an odd cycle in \(\ig{f}\).
	We show that every odd cycle of \(\ig{f'}\) has an internal \delt (*).
	Indeed, for every odd cycle \(C\) of \(\ig{f}\), we have two cases.
	\textbf{Case 1}: every internal \delt of \(C\) has the first vertex \(u\) such that \(m(u) \neq \star\).
	Then since \(u \in V(C)\), \(C\) is broken in \(\ig{f'}\) because all the input arcs of \(u\) in \(\ig{f}\) are removed in \(\ig{f'}\).
	\textbf{Case 2}: there is an internal \delt \((u, v_1, v_2)\) of \(C\) such that \(m(u) = \star\).
	If \(m(v_1) \neq \star\) or \(m(v_2) \neq \star\), then \(C\) is broken in \(\ig{f'}\) because \(v_1, v_2 \in V(C)\).
	Otherwise, the arcs \((uv_1, \oplus)\) and \((uv_2, \ominus)\) of \(\ig{f}\) are retained in \(\ig{f'}\).
	In this case, if \(C\) still appears in \(\ig{f'}\), \((u, v_1, v_2)\) is still an internal \delt of \(C\) in \(\ig{f'}\).
	In all cases, (*) still preserves.
	
	It follows that \(f'\) has at least one fixed point (say \(m_{fix}\)) by~\Cref{theo:AND-NOT-BN-fix-odd-cycle-internal-del-triple}.
	Since a fixed point is a two-valued complete trap space, \(m_{fix}(v) = m(v)\) for every \(v \in \var{f}, m(v) \neq \star\).
	Obviously, \(m_{fix}\) is a complete trap space of \(f\).
	We have \(m_{fix} <_s m\) which is a contradiction.
	Hence, all the \(\leq_s\)-minimal trap spaces of \(f\) are two-valued.
	This implies that all the regular models of \(P\) are two-valued.
\end{proof}

\Cref{cor:uni-rule-Datalog-tight-odd-cycle-internal-del-triple-StM} and~\Cref{theo:uni-rule-Datalog-tight-odd-cycle-internal-del-triple-RegM} are true for only tight \upnames.
We make two respective conjectures for (general) \upnames.

\begin{conjecture}\label{conj:uni-rule-Datalog-odd-cycle-internal-del-triple-StM}
	Let \(P\) be a \upname.
	If every odd cycle of \(\adg{P}\) has an internal \delt, then \(P\) has at least one stable model.
\end{conjecture}

\begin{conjecture}\label{conj:uni-rule-Datalog-odd-cycle-internal-del-triple-RegM}
	Let \(P\) be a \upname.
	If every odd cycle of \(\adg{P}\) has an internal \delt, then all the regular models of \(P\) are stable models.
\end{conjecture}

\begin{remark}
	If~\Cref{conj:uni-rule-Datalog-odd-cycle-internal-del-triple-StM} is true, then we can imply that~\Cref{conj:uni-rule-Datalog-odd-cycle-internal-del-triple-RegM} is true by applying the arguments of the proof of~\Cref{theo:uni-rule-Datalog-tight-odd-cycle-internal-del-triple-RegM}.
\end{remark}

\subsection{Model \unitextu and Even Cycles}\label{subsec:model-unicity-even-cycle}

In this subsection, we explore how the absence of even cycles in the atom dependency graph of a uni-rule \pname influences the uniqueness and structure of its two-valued or three-valued stable models.

\paragraph{General \DLN Programs.}

While uni-rule \pnames admit a clean correspondence between their atom dependency graphs and influence graphs of their associated \acbns, this connection does not carry over to general \pname.
Indeed, the characterization of complete trap spaces in \acbns fundamentally relies on the three-valued logic, which goes beyond what influence graphs can represent. 
As a result, we must abandon influence graphs and instead develop a new type of graphical representation—one that captures the syntactic dependencies and interactions of Boolean functions without relying on their dynamical interpretation.

\begin{definition}\label{def:BN-syntactic-IG}
	Given a \acbn \(f\), we define its \emph{syntactic influence graph} (denoted by \(\syng{f}\)) as follows:
	\begin{itemize}
		\item \((v_jv_i, \oplus)\) is an arc of \(\syng{f}\) \ifftext \(v_j\) appears in \(f_{v_i}\)
		\item \((v_jv_i, \ominus)\) is an arc of \(\syng{f}\) \ifftext \(\neg v_j\) appears in \(f_{v_i}\)
	\end{itemize}
\end{definition}

It is natural to obtain the coincidence between the atom dependency graph of a \pname and the syntactic influence graph of its encoded \acbn.

\begin{corollary}\label{cor:Datalog-adg-BN-sig}
	Given a \pname \(P\), let \(f\) be its encoded \acbn.
	Then \(\adg{P} = \syng{f}\).
\end{corollary}
\begin{proof}
	This immediately follows from the construction of the encoded \acbn and the definition of the syntactic influence graph.
\end{proof}

The following definition introduces several notations to capture input relationships in a signed directed graph, which we will use in subsequent results, including a key insight that establishes conditions under which a \acbn has a unique complete trap space.

\begin{definition}\label{def:graph-pos-neg-inputs}
	Consider a signed directed graph \(G\).
	Let \(\INpos{G}{v}\) denote the set of input vertices that have positive arcs to \(v\) in \(G\).
	Let \(\INneg{G}{v}\) denote the set of input vertices that have negative arcs to \(v\) in \(G\).
	We then define \(\INall{G}{v} = \INpos{G}{v} \cup \INneg{G}{v}\).
\end{definition}

\begin{lemma}\label{lem:BN-no-constant-no-even-cycle-unique-CoTS}
	Consider a \acbn \(f\).
	Assume that \(\syng{f}\) has the minimum in-degree of at least one.
	If \(\syng{f}\) has no even cycle, then \(f\) has a unique complete trap space.
\end{lemma}
\begin{proof}
	Since \(\syng{f}\) has the minimum in-degree of at least one, \(f\) has no variable \(v\) such that either \(f_v = \fval\) or \(f_v = \tval\).
	Then the sub-space \(\varepsilon\) where all variables are free is simply a complete trap space of \(f\).
	Assume that \(f\) has a complete trap space \(m \neq \varepsilon\).
	It follows that there is a variable \(v_0 \in \var{f}\) such that \(m(v_0) 
	\neq \uval\).
	We show that \(\syng{f}\) has an even cycle that can be constructed from \(v_0\) (*).
	
	Let \(\Pi_0\) be the subset of \(\INall{\syng{f}}{v_0}\) such that for every \(v \in \Pi_0\), both \(v \in \INpos{\syng{f}}{v_0}\) and \(v \in \INneg{\syng{f}}{v_0}\) hold.
	If there is a vertex \(v \in \Pi_0\) such that \(m(v) \neq \uval\), then we choose \(v_1\) as \(v\).
	Otherwise, we consider the set \(V_0 = \INall{\syng{f}}{v_0} \setminus \Pi_0\).
	The set \(V_0\) cannot be empty because if so, \(m(v) = \uval\) for every \(v \in \INall{\syng{f}}{v_0}\), leading to \(m(f_{v_0}) = \uval \neq m(v_0)\) and \(m\) cannot be a complete trap space of \(f\).
	Since \(m(v_0) \neq \uval\), we have two cases as follows.
	\textbf{Case 1}: \(m(v_0) = \fval\).
	Then \(m(f_{v_0}) = m(v_0) = \fval\) because \(m\) is a complete trap space of \(f\).
	If \(m(v) \neq \fval\) for every \(v \in \INpos{\syng{f}}{v_0} \cap V_0\) and \(m(v) \neq \tval\) for every \(v \in \INneg{\syng{f}}{v_0} \cap V_0\), then \(m(f_{v_0})\) is either \(\tval\) or \(\uval\) due to the three-valued logic semantics.
	This means that \(m(f_{v_0}) \neq \fval\) always holds in this case, which is a contradiction.
	Hence, there is a variable \(v \in V_0\) such that if \(v \in \INpos{\syng{f}}{v_0}\) then \(m(v) = \fval\) and if \(v \in \INneg{\syng{f}}{v_0}\) then \(m(v) = \tval\).
	We choose \(v_1\) as \(v\).
	\textbf{Case 2}: \(m(v_0) = \tval\).
	Then \(m(f_{v_0}) = m(v_0) = \tval\) because \(m\) is a complete trap space of \(f\).
	If \(m(v) \neq \tval\) for every \(v \in \INpos{\syng{f}}{v_0} \cap V_0\) and \(m(v) \neq \fval\) for every \(v \in \INneg{\syng{f}}{v_0} \cap V_0\), then \(m(f_{v_0})\) is either \(\fval\) or \(\uval\) due to the three-valued logic semantics.
	This means that \(m(f_{v_0}) \neq \tval\) always holds in this case, which is a contradiction.
	Hence, there is a variable \(v \in V_0\) such that if \(v \in \INpos{\syng{f}}{v_0}\) then \(m(v) = \tval\) and if \(v \in \INneg{\syng{f}}{v_0}\) then \(m(v) = \fval\).
	We choose \(v_1\) as \(v\).
	Following the two above cases, \(m(v_1)\) takes the value of \(m(v_0)\) if \(v_1 \in \INpos{\syng{f}}{v_0} \cap V_0\) and \(m(v_1)\) takes the value of \(\neg m(v_0)\) if \(v_1 \in \INneg{\syng{f}}{v_0} \cap V_0\) (**).
	Note that \(m(v_1) \neq \uval\) always holds.
	
	Repeating the above construction, we obtain an infinite descending chain \(v_0 \xleftarrow{s_0} v_1 \xleftarrow{s_1} v_2 \xleftarrow{s_2} \dots\) where for every \(i \leq 0\), \(v_i \in \var{f}\), \(m(v_i) \neq \uval\), and \(s_i\) is both \(\oplus\) and \(\ominus\) or \(s_i\) is either \(\oplus\) or \(\ominus\).
	Since \(\var{f}\) is finite, there are two integer numbers \(j\) and \(k\) (\(j, k \leq 0\)) such that \(v_{j} = v_{j + k}\).
	Let \(V = \{v_{j}, v_{j + 1}, \dots, v_{j + k}\}\).
	Since \(v_{j} = v_{j + k}\), \(\syng{f}[V]\) contains at least one cycle, and every its cycle is constituted by all vertices in \(V\).
	If there is \(i \in \{j, j + 1, \dots, j + k - 1\}\) such that \(s_i\) is both \(\oplus\) and \(\ominus\), then \(\syng{f}[V]\) contains both even and odd cycles.
	If \(s_i\) is either \(\oplus\) or \(\ominus\) for every \(i \in \{j, j + 1, \dots, j + k - 1\}\), then \(v_j \xleftarrow{s_j} v_{j + 1} \xleftarrow{s_{j + 1}} v_{j + 2} \dots \xleftarrow{s_{j + k - 1}} v_{j + k}\) is a cycle of \(\syng{f}\).
	By (**), the number of \(\ominus\) signs in this cycle must be even, thus this cycle is even.
	
	We have that (*) contradicts to the non-existence of even cycles in \(\syng{f}\).
	Hence, \(\varepsilon\) is the unique complete trap space of \(f\).
\end{proof}

\Cref{lem:BN-no-constant-no-even-cycle-unique-CoTS} requires the condition on the minimum in-degree.
To relax this condition, we introduce a process that iteratively eliminates syntatic constants from the \acbn. 
This begins with the notion of one-step syntatic percolation, defined as follows.

\begin{definition}\label{def:BN-one-step-syntatic-percolation}
	Consider a \acbn \(f\).
	A variable \(v \in \var{f}\) is called \emph{syntatic constant} if either \(f_v = 0\) or \(f_v = 1\).
	Let \(\sync{f}\) denote the set of syntatic constant variables of \(f\).
	We define the \emph{one-step syntatic percolation} of \(f\) (denoted as \(\onesynper{f}\)) as follows: \(\var{\onesynper{f}} = \var{f} \setminus \sync{f}\), and for every \(v \in \var{\onesynper{f}}\), \(\onesynper{f}_v = f'_v\), where \(f'_v\) is the Boolean function obtained by substituting syntatic constant values of \(f\)  in \(f_v\) with their Boolean functions \wrttext the three-valued logic.
\end{definition}

\begin{definition}\label{def:BN-full-syntatic-percolation}
	Consider a \acbn \(f\).
	The \emph{syntatic percolation} of \(f\) (denoted by \(\synper{f}\)) is obtained by applying the one-step syntatic percolation operator starting from \(f\) until it reaches a \acbn \(f'\) such that \(f' = \emptyset\) or \(\onesynper{f'} = f'\); this is always possible because the number of variables is finite.
\end{definition}

\begin{proposition}\label{prop:BN-syng-syntatic-percolation}
	Consider a \acbn \(f\).
	One-step syntatic percolation may reduce the set of variables and may reduce the set of arcs, i.e., \(V(\syng{\onesynper{f}}) \subseteq V(\syng{f})\) and \(E(\syng{\onesynper{f}}) \subseteq E(\syng{f})\).
	Consequently, \(V(\syng{\synper{f}}) \subseteq V(\syng{f})\) and \(E(\syng{\synper{f}}) \subseteq E(\syng{f})\).
\end{proposition}
\begin{proof}
	Note that \(\var{\onesynper{f}} = \var{f} \setminus \sync{f} \subseteq \var{f}\) by definition.
	Hence, \(V(\syng{\onesynper{f}}) \subseteq V(\syng{f})\) (*).
	Consider a variable \(v \in \var{\onesynper{f}}\).
	Let \(u\) be a literal that appears in \(\onesynper{f}_v\).
	Then \(u\) must appears in \(f_v\) because \(\onesynper{f}_v\) is obtained by syntatically simplifying \(f_v\).
	This implies that if \((uv, \oplus)\) is an arc in \(\syng{\onesynper{f}}\), it is an arc in \(\syng{f}\).
	Similarly, if \((uv, \ominus)\) is an arc in \(\syng{\onesynper{f}}\), it is an arc in \(\syng{f}\).
	Now, we can conclude that \(E(\syng{\onesynper{f}}) \subseteq E(\syng{f})\) (**).
	By applying (*) and (**) sequentially, we obtain that \(V(\syng{\synper{f}}) \subseteq V(\syng{f})\) and \(E(\syng{\synper{f}}) \subseteq E(\syng{f})\).
\end{proof}

\begin{proposition}\label{prop:BN-syntatic-percolation-CoTS}
	Given a \acbn \(f\), the set of complete trap spaces of \(\onesynper{f}\) one-to-one corresponds to the set of complete trap space of \(f\).
	Consequently, the set of complete trap spaces of \(\synper{f}\) one-to-one corresponds to the set of complete trap space of \(f\).
\end{proposition}
\begin{proof}
	Assume that \(m\) is a complete trap space of \(f\).
	For every \(v \in \var{f}\) such that either \(f_v = 0\) or \(f_v = 1\), \(m(v) = m(f_v) = f_v\).
	For every \(v \in \var{f}\) such that neither \(f_v = 0\) nor \(f_v = 1\), \(m(\onesynper{f}_v) = m(f_v) = m(v)\) because \(\onesynper{f}_v\) is obtained from \(f_v\) by substituting syntatic constant values of \(f\)  in \(f_v\) with their Boolean functions \wrttext the three-valued logic.
	Hence, the projection of \(m\) to \(\var{\onesynper{f}}\) is a complete trap space of \(\onesynper{f}\).
	
	Assume that \(m'\) is a complete trap space of \(\onesynper{f}\).
	Let \(m\) be a sub-space of \(f\) such that \(m(v) = f_v\) for every \(v \in \sync{f}\) and \(m(v) = m'(v)\) for every \(v \in \var{\onesynper{f}}\).
	For every \(v \in \sync{f}\), \(m(v) = f_v = m(f_v)\) because \(f_v \in \threed{}\).
	For every \(v \in \var{\onesynper{f}}\), \(m(v) = m'(v) = m(\onesynper{f}_v) = m(f_v)\) by the definition of \(\onesynper{f}_v\).
	This implies that \(m\) is a complete trap space of \(f\).
	
	Now we can conclude that the set of complete trap spaces of \(\onesynper{f}\) one-to-one corresponds to the set of complete trap space of \(f\).
	By applying this property sequentially, we obtain that the set of complete trap spaces of \(\synper{f}\) one-to-one corresponds to the set of complete trap space of \(f\).
\end{proof}

Now, we have enough ingredients to prove the following important result.

\begin{theorem}\label{theo:BN-no-even-cycle-unique-CoTS}
	Consider a \acbn \(f\).
	If \(\syng{f}\) has no even cycle, then \(f\) has a unique complete trap space.
\end{theorem}
\begin{proof}
	Let \(\synper{f}\) be the syntatic percolation of \(f\).
	By construction, there are two cases.
	\textbf{Case 1}: \(\synper{f} = \emptyset\).
	In this case, it is easy to see that \(f\) has a unique complete trap space specified by syntatic constant values of variables through the construction of \(\synper{f}\).
	\textbf{Case 2}: \(\synper{f} \neq \emptyset\) and \(\syng{\synper{f}}\) has the minimum in-degree of at least one.
	In this case, by~\Cref{prop:BN-syng-syntatic-percolation}, \(V(\syng{\synper{f}}) \subseteq V(\syng{f})\) and \(E(\syng{\synper{f}}) \subseteq E(\syng{f})\).
	Since \(\syng{f}\) has no even cycle, \(\syng{\synper{f}}\) has no even cycle.
	By~\Cref{lem:BN-no-constant-no-even-cycle-unique-CoTS}, \(\synper{f}\) has a unique complete trap space.
	By~\Cref{prop:BN-syntatic-percolation-CoTS}, the set of complete trap spaces of \(f\) one-to-one corresponds to the set of complete trap spaces of \(\synper{f}\).
	Hence, \(f\) has a unique complete trap space.
\end{proof}

We now show that the even-cycle freeness in the atom dependency graph of a \pname ensures not only the existence but also the \unitext of several fundamental semantic objects. Specifically, if \(\adg{P}\) has no even cycle, then \(P\) has a unique supported partial model, a unique stable partial model, and a unique regular model. 
Furthermore, this implies that \(P\) has at most one stable model. 
These results follow from the correspondence between atom dependency graphs and syntactic influence graphs.

\begin{theorem}[\textbf{main result}]\label{theo:Datalog-no-even-cycle-unique-SuPM}
	Let \(P\) be a \pname.
	If \(\adg{P}\) has no even cycle, then \(P\) has a unique supported partial model.
\end{theorem}
\begin{proof}
	Let \(f\) be the encoded \acbn of \(P\).
	By~\Cref{cor:Datalog-adg-BN-sig}, \(\syng{f} = \adg{P}\).
	Therefore, \(\syng{f}\) has no even cycle.
	By~\Cref{theo:BN-no-even-cycle-unique-CoTS}, \(f\) has a unique complete trap space.
	By~\Cref{theo:Datalog-SuPM-BN-CoTS}, \(P\) has a unique supported partial model.
\end{proof}

\begin{corollary}
  \label{cor:Datalog-no-even-cycle-unique-StPM}
	Let \(P\) be a \pname.
	If \(\adg{P}\) has no even cycle, then \(P\) has a unique stable partial model.
\end{corollary}
\begin{proof}
	By~\Cref{theo:Datalog-no-even-cycle-unique-SuPM}, \(P\) has a unique supported partial model.
	A stable partial model is also a supported partial model.
	Since \(P\) always has at least one stable partial model, it has a unique stable partial model.
\end{proof}

\begin{corollary}
  \label{cor:Datalog-no-even-cycle-unique-RegM}
	Let \(P\) be a \pname.
	If \(\adg{P}\) has no even cycle, then \(P\) has a unique regular model.
\end{corollary}
\begin{proof}
	This immediately follows from~\Cref{cor:Datalog-no-even-cycle-unique-StPM}, the existence of regular models, the fact that a regular model is a stable partial model.
\end{proof}

\begin{corollary}
  \label{cor:Datalog-no-even-cycle-at-most-one-StM}
	Let \(P\) be a \pname.
	If \(\adg{P}\) has no even cycle, then \(P\) has at most one stable model.
\end{corollary}
\begin{proof}
	This immediately follows from~\Cref{cor:Datalog-no-even-cycle-unique-StPM} and the fact that a stable model is a two-valued stable partial model.
	There exists some \pname such that its atom dependency graph has no even cycle and it has no stable model.
	For example, consider \pname \(P = \{p \leftarrow \dng{p}\}\).
	The atom dependency graph of \(P\) has no even cycle and \(P\) has no stable model.
\end{proof}

\paragraph{Uni-rule \DLN Programs.}

Similar to the case of odd cycles, by considering \upnames, we obtain two stronger results as follows.

\begin{theorem}[Theorem 2' of~\cite{RR2013}]\label{theo:AND-NOT-BN-no-strong-even-cycle-fix}
	Consider an \anbn \(f\).
	If every even cycle of \(\ig{f}\) has a \delt, then \(f\) has at most one fixed point.
\end{theorem}

\begin{theorem}\label{theo:uni-rule-Datalog-no-even-cycle-del-triple-at-most-one-StM}
	Let \(P\) be a \upname.
	If every even cycle of \(\adg{P}\) has a \delt, then \(P\) has at most one stable model.
\end{theorem}
\begin{proof}
	Let \(f\) be the encoded \acbn of \(P\).
	By~\Cref{prop:uni-rule-Datalog-2-AND-NOT-BN}, \(f\) is an \anbn.
	By~\Cref{prop:uni-rule-Datalog-AND-NOT-BN-graphs}, \(\ig{f} = \adg{P}\).
	It follows that every even cycle of \(\ig{f}\) has a \delt.
	By~\Cref{theo:AND-NOT-BN-no-strong-even-cycle-fix}, \(f\) has at most one fixed point.
	By~\Cref{cor:Datalog-SuM-BN-fix}, \(P\) has at most one stable model.
\end{proof}

\begin{theorem}[Lemma 1 of~\cite{TPRPA2025}]\label{theo:AND-NOT-BN-unique-async-att-no-strong-even-cycle}
	Consider an \anbn \(f\).
	If every even cycle of \(\ig{f}\) has a \delt, then \(\astg{f}\) has a unique attractor.
\end{theorem}

\begin{theorem}\label{theo:tight-uni-rule-Datalog-no-even-cycle-del-triple-unique-RegM}
	Let \(P\) be a tight \upname.
	If every even cycle of \(\adg{P}\) has a \delt, then \(P\) has a unique regular model.
\end{theorem}
\begin{proof}
	Let \(f\) be the encoded \acbn of \(P\).
	By~\Cref{prop:uni-rule-Datalog-2-AND-NOT-BN}, \(f\) is an \anbn.
	By~\Cref{prop:uni-rule-Datalog-AND-NOT-BN-graphs}, \(\ig{f} = \adg{P}\).
	It follows that every even cycle of \(\ig{f}\) has a \delt.
	Since \(P\) is tight, by~\Cref{lem:Datalog-tight-BN-RegM-s-min-TS}, the regular models of \(P\) coincide with the \(\leq_s\)-minimal trap spaces of \(f\).
	By~\Cref{theo:AND-NOT-BN-unique-async-att-no-strong-even-cycle}, \(\astg{f}\) has a unique attractor.
	Hence, \(f\) has a unique \(\leq_s\)-minimal trap space, leading to \(P\) has a unique regular model.
\end{proof}

\subsection{Number of Models and Feedback Vertex Sets}\label{subsec:num-model-fvs}

Understanding the number of semantic models a \pname can admit is crucial for analyzing its behavior, especially in applications involving reasoning, verification, or program synthesis~\citep{DT1996,CT1999,Linke2001,LZ2004,Costantini2006,FH2021}.
In this subsection, we investigate structural upper bounds on the number of supported partial models, stable partial models, stable models, and regular models of a \pname based on the size of its feedback vertex set.
These bounds provide a direct link between the combinatorial structure of the atom dependency graph and the semantic complexity of the program. 
In particular, they offer practical guidance for estimating or constraining the space of possible models, which is valuable for both theoretical studies and tool-supported analysis of \pnames.

\paragraph{General \DLN Programs.}

To the best of our knowledge, there is no existing work connecting between stable, stable partial, and regular models of a \pname and feedback vertex sets of its atom dependency graph.
We first relate the number of stable partial models and even feedback vertex sets (see~\Cref{theo:Datalog-num-SuPM-even-cycle}).
We state and prove the upper bound of \(3^{|U|}\) for the number of complete trap spaces in a \acbn (\Cref{theo:BN-num-CoTS-even-cycle}) where \(U\) is an even feedback vertex set of the syntatic dependency graph of this \acbn, then apply this bound to the number of stable partial models in a \pname.
The underlying intuition for the base of three is that in a stable partial model, the value of an atom can be \(\tval\), \(\fval\), or \(\uval\).
The underlying intuition for the exponent of \(|U|\) is that \(U\) interesects every even cycle and the \acbn has a unique complete trap space in the absence of even cycles (see~\Cref{theo:BN-no-even-cycle-unique-CoTS}).

\begin{theorem}\label{theo:BN-num-CoTS-even-cycle}
	Consider a \acbn \(f\).
	Let \(U\) be a subset of \(\var{f}\) that intersects every even cycle of \(\syng{f}\).
	Then the number of complete trap spaces of \(f\) is at most \(3^{|U|}\).
\end{theorem}
\begin{proof}
	Let \(I\) be an assignment \(U \mapsto \threed{}\).
	We build the new \acbn \(f^{I}\) as follows.
	For every \(v \in U\), \(f^{I}_v = I(v)\) if \(I(v) \neq \uval\), \(f^{I}_v = \neg v\) if \(I(v) = \uval\).
	For every \(v \in \var{f} \setminus U\), \(f^{I}_v = f_v\).
	Since \(U\) intersects every even cycle of \(\syng{f}\) and only negative self arcs can be introduced, \(\syng{f^{I}}\) has no even cycle.
	By~\Cref{theo:BN-no-even-cycle-unique-CoTS}, \(f^{I}\) has a unique complete trap space.
	By the construction, the setting \(f^{I}_v = I(v)\) (resp.\ \(f^{I}_v = \neg v\)) ensures that for any complete trap space of \(f^{I}\), the value of \(v\) is always \(I(v)\) (resp.\ \(\uval\)).
	Hence, this unique complete trap space agrees with the assignment \(I\).
	
	It is easy to see that a complete trap space of \(f\) that agrees with the assignment \(I\) is a complete trap space of \(f^{I}\).
	Since \(f^{I}\) has a unique complete trap space, we have an injection from the set of complete trap spaces of \(f\) to the set of possible assignments \(I\).
	There are in total \(3^{|U|}\) possible assignments \(I\).
	Hence, we can conclude that \(f\) has at most \(3^{|U|}\) complete trap spaces.
\end{proof}

\begin{theorem}\label{theo:Datalog-num-SuPM-even-cycle}
	Consider a \pname \(P\).
	Let \(U\) be a subset of \(\hb{P}\) that intersects every even cycle of \(\adg{P}\).
	Then the number of supported partial models of \(P\) is at most \(3^{|U|}\).
\end{theorem}
\begin{proof}
	Let \(f\) be the encoded \acbn of \(P\).
	By~\Cref{theo:Datalog-SuPM-BN-CoTS}, the supported partial models of \(P\) coincide with the complete trap spaces of \(f\).
	By~\Cref{cor:Datalog-adg-BN-sig}, \(\syng{f} = \adg{P}\), thus \(U\) is also a subset of \(\var{f}\) that intersects every even cycle of \(\syng{f}\).
	By~\Cref{theo:BN-num-CoTS-even-cycle}, \(f\) has at most \(3^{|U|}\) complete trap spaces.
	This implies that \(P\) has at most \(3^{|U|}\) supported partial models.
\end{proof}

The upper bound established in~\Cref{theo:Datalog-num-SuPM-even-cycle} not only constrains the number of supported partial models, but also extends naturally to more restrictive semantic notions, as shown in the following corollary.

\begin{corollary}
  \label{cor:Datalog-num-StPM-RegM-StM-even-cycle}
	Consider a \pname \(P\).
	Let \(U\) be a subset of \(\hb{P}\) that intersects every even cycle of \(\adg{P}\).
	Then the number of stable partial models of \(P\) is at most \(3^{|U|}\).
	In addition, this upper bound also holds the number of regular models or stable models.
\end{corollary}
\begin{proof}
	This immediately follows from~\Cref{theo:Datalog-num-SuPM-even-cycle} and the fact that a stable partial model is also a supported partial model, and a regular or stable model is also a stable partial model.
\end{proof}

\begin{remark}
	Let us consider~\Cref{exam:Datalog-all} again.
	The graph \(\adg{P}\) is given in~\Cref{fig:exam-Datalog-adg-tgst-tgsp}~(a).
	It is easy to verify that \(U = \{p\}\) or \(U = \{q\}\) intersects every even cycle of \(\adg{P}\).
	Hence,~\Cref{cor:Datalog-num-StPM-RegM-StM-even-cycle} gives the upper bound \(3^1 = 3\).
	Indeed, the \pname \(P\) has three stable (supported) partial models.
	Hence, the upper bound given by~\Cref{cor:Datalog-num-StPM-RegM-StM-even-cycle} can be reached.
\end{remark}

We then provide a simple upper bound for the number of regular models based on the connection to the dynamical behavior of a \acbn, in which the base decreases to 2 but the exponent increases to \(|\hb{P}|\).

\begin{proposition}\label{prop:Datalog-num-RegM-non-tight-trivial}
	Let \(P\) be a \pname.
	Then the number of regular models of \(P\) is at most \(2^{|\hb{P}|}\).
\end{proposition}
\begin{proof}
	Let \(\lfp{P}\) be the least fixpoint of \(P\) and \(f'\) be the encoded \acbn of \(\lfp{P}\).
	By~\Cref{lem:Datalog-lfp-BN-RegM-s-min-TS}, the regular models of \(P\) coincide with the \(\leq_s\)-minimal trap spaces of \(f'\).
	We have \(\hb{P} = \hb{\lfp{P}} = \var{f'}\) by definition.
	Since each attractor of \(\astg{f'}\) contains at least one state and \(f'\) has \(2^{|\var{f'}|}\) states, \(\astg{f'}\) has at most \(2^{|\var{f'}|}\) attractors.
	The number of \(\leq_s\)-minimal trap spaces of \(f'\) is a lower bound for the number of attractors of \(\astg{f'}\)~\citep{KBS2015}.
	Hence, \(P\) has at most \(2^{|\hb{P}|}\) regular models.
\end{proof}

We get a better upper bound for the number of regular models (see~\Cref{theo:Datalog-num-RegM-tight-even-cycle}), but restricted to tight \pnames.

\begin{theorem}[Corollary 2 of~\cite{R2009}]\label{theo:BN-num-asyn-att-even-cycle}
	Given a \acbn \(f\), let \(U\) be a subset of \(\var{f}\) that intersects every even cycle of \(\ig{f}\).
	Then the number of attractors of \(\astg{f}\) is at most \(2^{|U|}\).
\end{theorem}

\begin{theorem}
  \label{theo:Datalog-num-RegM-tight-even-cycle}
	Let \(P\) be a \pname.
	Let \(U\) be a subset of \(\hb{P}\) that intersects every even cycle of \(\adg{P}\).
	If \(P\) is tight, then the number of regular models of \(P\) is at most \(2^{|U|}\).
\end{theorem}
\begin{proof}
	Let \(f\) be the encoded \acbn of \(P\).
	By~\Cref{prop:Datalog-BN-graphs}, \(\ig{f}\) is a sub-graph of \(\adg{P}\), thus \(U\) is a subset of \(\var{f}\) that intersects every even cycle of \(\ig{f}\).
	By~\Cref{theo:BN-num-asyn-att-even-cycle}, \(\astg{f}\) has at most \(2^{|U|}\) attractors.
	The number of \(\leq_s\)-minimal trap spaces of \(f\) is a lower bound for the number of attractors of \(\astg{f}\)~\citep{KBS2015}.
	Hence, \(f\) has at most \(2^{|U|}\) \(\leq_s\)-minimal trap spaces.
	Since \(P\) is tight, the regular models of \(P\) coincide with the \(\leq_s\)-minimal trap spaces of \(f\) by~\Cref{lem:Datalog-tight-BN-RegM-s-min-TS}.
	This implies that \(P\) has at most \(2^{|U|}\) regular models.
\end{proof}

\begin{remark}
	We here analyze the three above upper bounds for the number of regular models in a \pname.
	For tight \pnames, the bound \(2^{|U|}\) is the best.
	However, for non-tight \pnames, it is not applicable.
	The set \(U\) is always smaller (even much smaller in most cases) than or equal to the set \(\hb{P}\), yet \(3^{|U|}\) is not always smaller than \(2^{|\hb{P}|}\).
	Hence the bound \(2^{|\hb{P}|}\) still has merit for non-tight \pnames.
	Let us consider~\Cref{exam:Datalog-all} again.
	We have \(\hb{P} = \{p, q, r\}\).
	The graph \(\adg{P}\) is given in~\Cref{fig:exam-Datalog-adg-tgst-tgsp}~(a).
	It is easy to verify that \(P\) is tight and \(U = \{p\}\) or \(U = \{q\}\) intersects every even cycle of \(\adg{P}\).
	Hence,~\Cref{cor:Datalog-num-StPM-RegM-StM-even-cycle} gives the upper bound \(3^1 = 3\),~\Cref{prop:Datalog-num-RegM-non-tight-trivial} gives the upper bound \(2^3 = 8\), whereas~\Cref{theo:Datalog-num-RegM-tight-even-cycle} gives the upper bound of \(2^1 = 2\).
	Indeed, the \pname \(P\) has two regular models.
	Hence, the upper bound given by~\Cref{theo:Datalog-num-RegM-tight-even-cycle} can be reached \wrttext tight \pnames.
\end{remark}

Since a stable model is a regular model, we can derive from~\Cref{theo:Datalog-num-RegM-tight-even-cycle} that \(2^{|U|}\) is also an upper bound for the number of stable models of a tight \pname.
However, we prove that this bound also holds for non-tight \pnames (see~\Cref{theo:Datalog-num-StM-general-even-cycle}).

\begin{theorem}[Corollary 10 of~\cite{Aracena2008}]\label{theo:BN-num-fix-even-cycle}
	Given a \acbn \(f\), let \(U\) be a subset of \(\var{f}\) that intersects every even cycle of \(\ig{f}\).
	Then the number of fixed points of \(f\) is at most \(2^{|U|}\).
\end{theorem}

\begin{theorem}
  \label{theo:Datalog-num-StM-general-even-cycle}
	Consider a \pname \(P\).
	Let \(U\) be a subset of \(\hb{P}\) that intersects every even cycle of \(\adg{P}\).
	Then the number of stable models of \(P\) is at most \(2^{|U|}\).
\end{theorem}
\begin{proof}
	Let \(f\) be the encoded \acbn of \(P\).
	By~\Cref{prop:Datalog-BN-graphs}, \(\ig{f}\) is a sub-graph of \(\adg{P}\), thus \(U\) is a subset of \(\var{f}\) that intersects every even cycle of \(\ig{f}\).
	By~\Cref{theo:BN-num-fix-even-cycle}, \(f\) has at most \(2^{|U|}\) fixed points.
	By~\Cref{cor:Datalog-SuM-BN-fix} and the fact that a stable model is a supported model, \(P\) has at most \(2^{|U|}\) stable models.
\end{proof}

\begin{remark}\label{remark:upper-bound-StM}
	We here recall some existing upper bounds for the number of stable models in \pnames.
	Given a \pname \(P\),~\citet{CT1999} proved the upper bound of \(3^{n/3}\) where \(n\) is the number of rules in \(\gr{P}\).
	\citet{LZ2004} later proved the upper bound of \(2^k\) where \(k\) is the number of even cycles in \(\adg{P}\).
	Our new result (i.e.,~\Cref{theo:Datalog-num-StM-general-even-cycle}) provides the upper bound of \(2^{|U|}\) where \(U\) is an even feedback vertex set of \(\adg{P}\).
	It is easy to see that we always find an even feedback vertex set \(U\) such that \(|U| \leq k\), showing that our result is more general than that of~\cite{LZ2004}.
	It is however hard to directly compare between \(2^{|U|}\) and \(3^{n/3}\).
	Consider the (ground) \pname \(\{a \leftarrow \dng{b}; a \leftarrow a; b \leftarrow \dng{a}\}\).
	Its atom dependency graph is given in~\Cref{fig:exam-upper-bound-StM}.
	It is easy to see that this graph has two even cycles: \(a \xrightarrow{\oplus} a\) and \(a \xrightarrow{\ominus} b \xrightarrow{\ominus} a\).
	Hence, the \pname is non-tight and \(U = \{a\}\) intersects every even cycle of the graph.
	The result of~\citet{CT1999} gives the upper bound \(3^{3/3} = 3\) for the number of stable models.
	The result of~\cite{LZ2004} gives the upper bound \(2^2 = 4\).
	Our new result gives the upper bound \(2^1 = 2\).
	Furthermore, we can see that the upper bound \(2^{|U|}\) can be reached.
\end{remark}

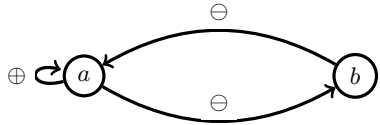
\begin{figure}[!ht]
	\centering
	\begin{tikzpicture}[node distance=2cm and 2cm, every node/.style={scale=1.0}]
		\node[circle, draw] (a) [] {$a$};
		\node[circle, draw] (b) [right=of a,xshift=1cm] {$b$};
		
		\draw[->] (a) edge [bend right=30] node [midway, above, fill=white] {$\ominus$} (b);
		\draw[->] (b) edge [bend right=30] node [midway, above, fill=white] {$\ominus$} (a);
		
		\draw[->] (a) edge [loop left] node [midway, left, fill=white] {$\oplus$} (a);
	\end{tikzpicture}
	\caption{Atom dependency graph of the \pname of~\Cref{remark:upper-bound-StM}.}\label{fig:exam-upper-bound-StM}
\end{figure}

Inspired by~\Cref{theo:Datalog-num-StM-general-even-cycle}, we make the following conjecture on the number of regular models in a (tight or non-tight) \pname.

\begin{conjecture}\label{conj:Datalog-num-RegM-general-even-cycle}
	Consider a \pname \(P\).
	Let \(U\) be a subset of \(\hb{P}\) that intersects every even cycle of \(\adg{P}\).
	Then the number of regular models of \(P\) is at most \(2^{|U|}\).
\end{conjecture}

\paragraph{Uni-rule \DLN Programs.} By considering \upnames, we obtain two tighter upper bounds for the numbers of stable and regular models.

\begin{theorem}[Theorem 3.5 of~\cite{VBHKWL12}]\label{theo:AND-NOT-BN-num-fix-cycle-delo-triple}
	Given an \anbn \(f\),
	let \(U\) be a subset of \(\var{f}\) that intersects every \deltfree even cycle of \(\ig{f}\).
	Then the number of fixed points of \(f\) is at most \(2^{|U|}\).
\end{theorem}

\begin{theorem}\label{theo:uni-rule-Datalog-num-StM-even-cycle-del-triple}
	Let \(P\) be a \upname.
	Assume that \(U\) is a subset of \(\hb{P}\) that intersects every even cycle without a \delt of \(\adg{P}\).
	Then \(P\) has at most \(2^{|U|}\) stable models.
\end{theorem}
\begin{proof}
	Let \(f\) be the encoded \acbn of \(P\).
	By~\Cref{prop:uni-rule-Datalog-2-AND-NOT-BN}, \(f\) is a \anbn.
	By~\Cref{prop:uni-rule-Datalog-AND-NOT-BN-graphs}, \(\ig{f} = \adg{P}\), thus \(U\) is a subset of \(\var{f}\) that intersects every \deltfree even cycle of \(\ig{f}\).
	By~\Cref{theo:AND-NOT-BN-num-fix-cycle-delo-triple}, \(f\) has at most \(2^{|U|}\) fixed points.
	By~\Cref{cor:Datalog-SuM-BN-fix} and the fact that a stable model is a supported model, \(P\) has at most \(2^{|U|}\) stable models.
\end{proof}

\begin{theorem}[Theorem 3 of~\cite{TPRPA2025}]\label{theo:AND-NOT-BN-num-async-att-strong-even-cycle}
	Given an \anbn \(f\), let \(U\) be a subset of \(\var{f}\) that intersects every \deltfree even cycle of \(\ig{f}\).
	Then the number of attractors of \(\astg{f}\) is at most \(2^{|U|}\).
\end{theorem}

\begin{theorem}
  \label{theo:uni-rule-Datalog-num-RegM-tight-even-cycle-del-triple}
	Let \(P\) be a \upname.
	Assume that \(U\) is a subset of \(\hb{P}\) that intersects every \deltfree even cycle of \(\adg{P}\).
	If \(P\) is tight, then \(P\) has at most \(2^{|U|}\) regular models.
\end{theorem}
\begin{proof}
	Let \(f\) be the encoded \acbn of \(P\).
	By~\Cref{prop:uni-rule-Datalog-2-AND-NOT-BN}, \(f\) is a \anbn.
	By~\Cref{prop:uni-rule-Datalog-AND-NOT-BN-graphs}, \(\ig{f} = \adg{P}\), thus \(U\) is a subset of \(\var{f}\) that intersects every \deltfree even cycle of \(\ig{f}\).
	By~\Cref{theo:AND-NOT-BN-num-async-att-strong-even-cycle}, \(\astg{f}\) has at most \(2^{|U|}\) attractors.
	The number of \(\leq_s\)-minimal trap spaces of \(f\) is a lower bound for the number of attractors of \(\astg{f}\)~\citep{KBS2015}.
	Hence, \(f\) has at most \(2^{|U|}\) \(\leq_s\)-minimal trap spaces.
	Since \(P\) is tight, the regular models of \(P\) coincide with the \(\leq_s\)-minimal trap spaces of \(f\) by~\Cref{lem:Datalog-tight-BN-RegM-s-min-TS}.
	This implies that \(P\) has at most \(2^{|U|}\) regular models.
\end{proof}

\begin{example}\label{exam:uni-rule-Datalog-num-RegM-StM}
	Consider the \pname \(P = \{v_1 \leftarrow \dng{v_2}; v_2 \leftarrow v_1; v_3 \leftarrow v_1, \dng{v_4}; v_4 \leftarrow \dng{v_1}, \dng{v_3}\}\).
	We use ``;'' to separate program rules.
	The atom dependency graph \(\adg{P}\) is shown in~\Cref{fig:uni-rule-Datalog-num-RegM-StM}.
	It is easy to verify that \(P\) is uni-rule and tight.
	The graph \(\adg{P}\) has only one even cycle \(C = v_3 \xrightarrow{\ominus} v_4 \xrightarrow{\ominus} v_3\).
	Then~\Cref{theo:Datalog-num-StM-general-even-cycle} (resp.\ \Cref{theo:Datalog-num-RegM-tight-even-cycle}) gives an upper bound \(2^1 = 2\) for the number of stable models (resp.\ regular models) of \(P\).
	However, \((v_1, v_3, v_4)\) is a \delt of \(C\).
	Then~\Cref{theo:uni-rule-Datalog-num-StM-even-cycle-del-triple} (resp.\ \Cref{theo:uni-rule-Datalog-num-RegM-tight-even-cycle-del-triple}) gives the upper bound \(2^0 = 1\) for the number of stable models (or regular models) of \(P\).
	Indeed, \(P\) has only one regular model \(\{v_1 = \uval, v_2 = \uval, v_3 = \uval, v_4 = \uval\}\) and no stable model.
\end{example}

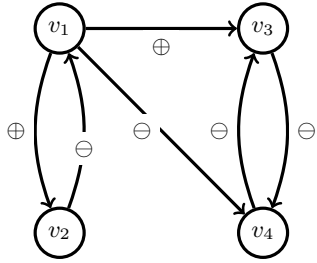
\begin{figure}[!ht]
	\centering
	\begin{tikzpicture}[node distance=2cm and 2cm]
		\node[circle, draw] (v1) [] {$v_1$};
		\node[circle, draw] (v2) [below=of v1, xshift=0cm] {$v_2$};
		\node[circle, draw] (v3) [right=of v1, xshift=0cm] {$v_3$};
		\node[circle, draw] (v4) [right=of v2, xshift=0cm] {$v_4$};
		
		\draw[->] (v1) edge [bend right=20] node [midway, left, fill=white] {$\oplus$} (v2);
		
		\draw[->] (v1) edge [] node [midway, left, fill=white] {$\ominus$} (v4);
		\draw[->] (v1) edge [] node [midway, below, fill=white] {$\oplus$} (v3);
		
		\draw[->] (v2) edge [bend right=20] node [midway, below, fill=white] {$\ominus$} (v1);
		
		\draw[->] (v3) edge [bend left=20] node [midway, right, fill=white] {$\ominus$} (v4);
		
		\draw[->] (v4) edge [bend left=20] node [midway, left, fill=white] {$\ominus$} (v3);
	\end{tikzpicture}
	\caption{Atom dependency graph of the uni-rule \pname \(P\) of~\Cref{exam:uni-rule-Datalog-num-RegM-StM}.}
	\label{fig:uni-rule-Datalog-num-RegM-StM}
\end{figure}

Inspired by~\Cref{theo:uni-rule-Datalog-num-RegM-tight-even-cycle-del-triple}, we make the following conjecture on the number of regular models in (tight or non-tight) uni-rule \pnames.

\begin{conjecture}\label{conj:uni-rule-Datalog-num-RegM-general-even-cycle}
	Let \(P\) be a \upname.
	Assume that \(U\) is a subset of \(\hb{P}\) that intersects every \deltfree even cycle of \(\adg{P}\).
	Then \(P\) has at most \(2^{|U|}\) regular models.
\end{conjecture}

\section{Trap Spaces for \DLN Programs}\label{sec:Datalog-trap-spaces}

In this section, we introduce the notions of \emph{stable trap space} and \emph{supported trap space} for \pnames, borrowed from the notion of trap spaces in \acbns. 
These constructs offer a new perspective for analyzing the model-theoretic and dynamical behavior of \pnames. 
We develop their basic properties, and establish formal relationships with classical semantics such as stable (supported) partial models, regular models, stable (supported) models, and stable (supported) classes. 
This unified view lays the foundation for leveraging trap space techniques in analysis and reasoning tasks involving \pnames.

\subsection{Definitions}\label{subsec:Datalog-trap-space-defs}

We begin by formally defining the central notions of stable and supported trap sets, which characterize non-empty sets of two-valued interpretations that are closed under the program's update operators.

\begin{definition}\label{def:Datalog-trap-set}
	A non-empty set \(S\) of two-valued interpretations of a \pname \(P\) is called a \emph{stable trap set} (resp.\ \emph{supported trap set}) of \(P\) if \(\{F_P(I) | I \in S\} \subseteq S\) (resp.\ \(\{T_P(I) | I \in S\} \subseteq S\)).
\end{definition}

Note that a stable (resp.\ supported) class is a stable (resp.\ supported) trap set, but the reverse may not be true.
Given the \pname \(P\) of~\Cref{exam:Datalog-all}, \(\{\{p, r\}, \{p\}\}\) is a stable (resp.\ supported) trap set of \(P\), but it is not a stable (resp.\ supported) class of \(P\).

\begin{definition}\label{def:Datalog-trap-space}
	A three-valued interpretation \(I\) of a \pname \(P\) is called a \emph{stable trap space} (resp.\ \emph{supported trap space}) of \(P\) if \(\cset{I}\) is a stable (resp.\ supported) trap set of \(P\).
\end{definition}

It is easy to adapt the concept of stable or supported trap set for a directed graph in general (see~\Cref{def:DiG-trap-set}).

\begin{definition}\label{def:DiG-trap-set}
	Consider a directed graph \(G\).
	A subset \(S\) of \(V(G)\) is called a \emph{trap set} of \(G\) if there are no two vertices \(A\) and \(B\) such that \(A \in S\), \(B \not \in S\), and \((A, B) \in E(G)\).
\end{definition}

It is easy to see that \(S\) is a stable (resp.\ supported) trap set of \(P\) iff \(S\) is a trap set of \(\tgst{P}\) (resp.\ \(\tgsp{P}\)).
Hence, we can deduce from~\Cref{def:Datalog-trap-space} that a three-valued interpretation \(I\) is a stable (resp.\ supported) trap space of \(P\) if \(\cset{I}\) is a trap set of \(\tgst{P}\) (resp.\ \(\tgsp{P}\)).
Since stable and supported transition graphs represent the dynamical aspect of a \pname~\citep{BS1992,IS2012}, this indicates that trap spaces represent the dynamical aspect of a \pname.

We now illustrate the notions of stable and supported trap spaces through a concrete example, demonstrating how they can be identified respectively from the stable and supported transition graphs of a given program.

\begin{example}\label{exam:Datalog-trap-spaces}
	Consider the \pname \(P\) of~\Cref{exam:Datalog-all}.
	\Cref{fig:exam-Datalog-adg-tgst-tgsp}~(b) and \Cref{fig:exam-Datalog-adg-tgst-tgsp}~(c) show the stable and supported transition graphs of \(P\), respectively.
	Then \(I_1 = \{p = \tval, q = \fval, r = \uval\}\) is a stable (resp.\ supported) trap space of \(P\) because \(\cset{I_1} = \{\{p\}, \{p, r\}\}\) is a trap set of \(\tgst{P}\) (resp.\ \(\tgsp{P}\)).
	By checking the remaining three-valued interpretations, we get the four other stable trap spaces that are also supported trap spaces of \(P\):
	\begin{align*}
		I_2 &= \{p = \fval, q = \tval, r = \uval\} \quad (\cset{I_2} = \{\{q\}, \{q, r\}\}),\\
		I_3 &= \{p = \uval, q = \uval, r = \uval\} \quad (\cset{I_3} = \{\{p\}, \{p, r\}, \{q\}, \{q, r\}, \{p, q\}, \{r\}, \{p, q, r\}, \emptyset\}),\\
		I_4 &= \{p = \tval, q = \fval, r = \fval\} \quad (\cset{I_4} = \{\{p\}\}),\\
		I_5 &= \{p = \fval, q = \tval, r = \tval\} \quad (\cset{I_5} = \{\{q, r\}\}).
	\end{align*}
\end{example}

\subsection{Properties}\label{subsec:Datalog-trap-space-properties}

In this subsection, we present basic properties of stable and supported trap spaces in \pnames. 
We begin by establishing their guaranteed existence, which ensures that the trap space framework is broadly applicable to the analysis of program dynamics. 
Further properties shall clarify their intrinsic characteristics.

\begin{proposition}\label{prop:Datalog-exist-StTS-SuTS}
	A \pname \(P\) always has a stable or supported trap space.
\end{proposition}
\begin{proof}
	Let \(I\) be a three-valued interpretation that corresponds to all the two-valued interpretations, i.e., \(\forall a \in \hb{P}, I(a) = \uval\).
	By setting \(S = \cset{I}\), the condition \(\{F_P(J) | J \in S\} \subseteq S\) (resp.\ \(\{T_P(J) | J \in S\} \subseteq S\)) always holds.
	Hence, \(\cset{I}\) is a stable (resp.\ supported) trap set of \(P\).
	By definition, \(I\) is a stable (resp.\ supported) trap space of \(P\).
\end{proof}

We then introduce an important concept, namely \emph{consistent}, on trap spaces of \pnames.
Two three-valued interpretations \(I_1\) and \(I_2\) are called \emph{consistent} if for all \(a \in \hb{P}\), \(I_1(a) \leq_s I_2(a)\) or \(I_2(a) \leq_s I_1(a)\).
Equivalently, \(I_1\) and \(I_2\) are called consistent if \(\cset{I_1} \cap \cset{I_2} \neq \emptyset\).
Note that using \(I_1 \leq_s I_2\) or \(I_2 \leq_s I_1\) is insufficient here, since there exist two three-valued interpretations that are not comparable \wrttext \(\leq_s\) but consistent.
When \(I_1\) and \(I_2\) are consistent, their overlap (denoted by \(I_1 \sqcap I_2\)) is a three-valued interpretation \(I\) such that for all \(a \in \hb{P}\), \(I(a) = \mins(I_1(a), I_2(a))\).
It also follows that \(\cset{I} = \cset{I_1} \cap \cset{I_2}\).

This notion of consistency enables us to study how trap spaces interact and combine, particularly through their common overlap.

\begin{proposition}\label{prop:Datalog-overlap-two-consistent-TS}
	Let \(P\) be a \pname.
	The overlap of two consistent stable (resp.\ supported) trap spaces of \(P\) is a stable (resp.\ supported) trap space of \(P\).
\end{proposition}
\begin{proof}
	Hereafter, we prove the case of stable trap spaces. 
	The proof for the case of supported trap spaces is symmetrical.
	Let \(I_1\) and \(I_2\) be two consistent stable trap spaces of \(P\).
	Then \(I_1 \sqcap I_2\) is also a three-valued interpretation by construction.
	Let \(s\) be an arbitrary two-valued interpretation in \(\cset{I_1 \sqcap I_2}\) and \(s'\) be its successor in \(\tgst{P}\).
	Since \(\cset{I_1 \sqcap I_2} = \cset{I_1} \cap \cset{I_2}\), \(s \in \cset{I_1}\) and \(s \in \cset{I_2}\).
	Since \(I_1\) (resp.\ \(I_2\)) is a stable trap space, \(s' \in \cset{I_1}\) (resp.\ \(s' \in \cset{I_2}\)).
	It follows that \(s' \in \cset{I_1 \sqcap I_2}\).
	Hence, \(\cset{I_1 \sqcap I_2}\) is a stable trap set, leading to \(I_1 \sqcap I_2\) is a stable trap space of \(P\).
\end{proof}

The above property of stable and supported trap spaces is analogous to a property of stable and supported classes in \pnames~\citep{IS2012}.
However, while the union of two stable (resp.\ supported) classes is a stable (resp.\ supported) class, that of two stable (resp.\ supported) trap spaces may not be a stable (resp.\ supported) trap space.
In~\Cref{exam:Datalog-trap-spaces}, \(\cset{I_4} \cup \cset{I_5}\) is a stable class of \(P\), whereas \(\cset{I_4} \cup \cset{I_5}\) does not correspond to any three-valued interpretation.

The closure of stable or supported trap spaces under consistent overlap leads to a useful structural consequence concerning the minimal trap space that covers a given set of two-valued interpretations.

\begin{corollary}\label{cor:Datalog-unique-minimal-covered-TS}
	Let \(P\) be a \pname.
	Let \(S\) be a non-empty set of two-valued interpretations of \(P\).
	Then there is a unique \(\leq_s\)-minimal stable (resp.\ supported) trap space, denoted by \(\uniCovStTS{S}\) (resp.\ \(\uniCovSuTS{S}\)), such that \(S \subseteq \cset{\uniCovStTS{S}}\) (resp.\ \(S \subseteq \cset{\uniCovSuTS{S}}\)).
\end{corollary}
\begin{proof}
	Hereafter, we prove the case of stable trap spaces. 
	The proof for the case of supported trap spaces is symmetrical.
	Let \(T\) be the set of all stable trap spaces \(I\) such that \(\cset{I}\) contains \(S\).
	We have \(T\) is non-empty as at least it contains \(\epsilon\) that is a stable trap space in which all atoms are assigned to \(\uval\).
	The elements in \(T\) are mutually consistent, thus we can take the overlap of all these elements (denoted by \(\uniCovStTS{S}\)).
	By applying the similar reasoning as in the proof of~\Cref{prop:Datalog-overlap-two-consistent-TS}, \(\uniCovStTS{S}\) is a stable trap space of \(P\).
	By construction, \(\uniCovStTS{S}\) is unique and \(\leq_s\)-minimal, and \(S \subseteq \cset{\uniCovStTS{S}}\).
\end{proof}

An another consequence of the closure of stable or supported trap spaces under consistent overlap is that two distinct \(\leq_s\)-minimal trap spaces cannot be consistent with each other, as formalized below.

\begin{proposition}\label{prop:Datalog-two-min-s-TS-inconsistent}
	Let \(P\) be a \pname.
	Let \(I_1\) and \(I_2\) be two distinct \(\leq_s\)-minimal stable (resp.\ supported) trap spaces of \(P\).
	Then \(I_1\) and \(I_2\) are not consistent.
\end{proposition}
\begin{proof}
	Hereafter, we prove the case of stable trap spaces. 
	The proof for the case of supported trap spaces is symmetrical.
	Assume that \(I_1\) and \(I_2\) are consistent.
	Then \(I = I_1 \sqcap I_2\) exists.
	By~\Cref{prop:Datalog-overlap-two-consistent-TS}, \(I\) is a stable trap space of \(P\).
	Since \(I_1\) and \(I_2\) are distinct, \(I <_s I_1\) or \(I <_s I_2\) must hold.
	This is a contradiction because \(I_1\) and \(I_2\) are \(\leq_s\)-minimal stable trap spaces of \(P\).
	Hence, \(I_1\) and \(I_2\) are not consistent.
\end{proof}

We now recall several important theoretical results from~\cite{IS2012} that lead us to similar results for trap spaces in \pnames.
Specifically, for negative \pnames, the stable and supported semantics coincide, resulting in identical transition graphs and thus identical sets of trap spaces; and the stable transition graph of an arbitrary \pname remains invariant under the least fixpoint transformation, which in turn implies that the set of stable trap spaces is also preserved under this transformation.

\begin{theorem}[Proposition 5.2 of~\cite{IS2012}]\label{theo:neg-Datalog-StTG-SuTG}
	Let \(P\) be a negative \pname.
	Then \(\tgst{P} = \tgsp{P}\), i.e., the stable and supported transition graphs of \(P\) are the same.
\end{theorem}

\begin{corollary}\label{cor:neg-Datalog-StTS-SuTS}
	Let \(P\) be a negative \pname.
	Then the set of stable trap spaces of \(P\) coincides with the set of supported trap spaces of \(P\).
\end{corollary}
\begin{proof}
	This immediately follows from~\Cref{theo:neg-Datalog-StTG-SuTG} and the dynamical characterizations of stable and supported trap spaces of a \pname.
\end{proof}

\begin{theorem}[Theorem 5.5 of~\cite{IS2012}]\label{theo:Datalog-lfp-StTG}
	Let \(P\) be a \pname and \(\lfp{P}\) denote the least fixpoint of \(P\).
	Then \(\tgst{P} = \tgst{\lfp{P}}\), i.e., \(P\) and \(\lfp{P}\) have the same stable transition graph.
\end{theorem}

\begin{corollary}\label{cor:Datalog-lfp-StTS}
	Let \(P\) be a \pname and \(\lfp{P}\) denote the least fixpoint of \(P\).
	Then the set of stable trap spaces of \(P\) coincides with the set of stable trap spaces of \(\lfp{P}\).
\end{corollary}
\begin{proof}
	This immediately follows from~\Cref{theo:Datalog-lfp-StTG} and the dynamical characterization of stable trap spaces of a \pname.
\end{proof}

\subsection{Relationships with Other Semantics}\label{subsec:trap-space-relationships}

Naturally, if a stable (resp.\ supported) trap space of a \pname is two-valued, then it is also a stable (resp.\ supported) model of this program.
Hereafter, we show more relationships between stable and supported trap spaces and other types of models in \pnames.
They highlight the role of trap spaces as generalizations of models, capturing both dynamically and semantically meaningful behavior of a program.
Understanding these relationships not only deepens our theoretical insight into the semantic landscape of \pnames but also may open the door to new algorithmic strategies.

\subsubsection{Stable and Supported Class Semantics}

We now investigate the connection between trap spaces and class-based semantics of \pnames, where models are understood in terms of cyclic or recurrent behavior in the transition graphs.

\begin{proposition}\label{prop:Datalog-min-trap-set-cycle-TS}
	Consider a \pname \(P\).
	A non-empty set of two-valued interpretations \(S\) is a \(\subseteq\)-minimal stable (resp.\ supported) trap set of \(P\) iff \(S\) forms a simple cycle of \(\tgst{P}\) (resp.\ \(\tgsp{P}\)).
\end{proposition}
\begin{proof}
	We have that each vertex in \(\tgst{P}\) (resp.\ \(\tgsp{P}\)) has exactly one on-going arc.
	They are similar to the synchronous state transition graph of a \acbn.
	
	Hence, \(S\) forms a simple cycle of \(\tgst{P}\) (resp.\ \(\tgsp{P}\)) \\
	iff \(S\) is a \(\subseteq\)-minimal stable (resp.\ supported) trap set of \(\tgst{P}\) (resp.\ \(\tgsp{P}\)) (see~\cite{EM2011}) \\
	iff \(S\) is a \(\subseteq\)-minimal stable (resp.\ supported) trap set of \(P\).
\end{proof}

We recall the two established results that precisely characterize strict stable and supported classes of a \pname in terms of simple cycles in the respective transition graphs.

\begin{proposition}[Theorem 3 of~\cite{BS1992}]\label{prop:Datalog-strict-StC-cycle-StTS}
	Consider a \pname \(P\).
	A non-empty set of two-valued interpretations \(S\) is a strict stable class of \(P\) iff \(S\) forms a simple cycle of \(\tgst{P}\).
\end{proposition}

\begin{proposition}[Theorem 3.2 of~\cite{IS2012}]\label{prop:Datalog-strict-SuC-cycle-SuTS}
	Consider a \pname \(P\).
	A non-empty set of two-valued interpretations \(S\) is a strict supported class of \(P\) iff \(S\) forms a simple cycle of \(\tgsp{P}\).
\end{proposition}

We then connect minimal trap sets and class semantics by showing that \(\subseteq\)-minimal stable and supported trap sets coincide with strict stable and supported classes, respectively.

\begin{proposition}\label{prop:Datalog-min-trap-set-strict-C}
	Consider a \pname \(P\).
	A non-empty set of two-valued interpretations \(S\) is a \(\subseteq\)-minimal stable (resp.\ supported) trap set of \(P\) iff \(S\) is a strict stable (resp.\ supported) class of \(P\).
\end{proposition}
\begin{proof}
	We show the proof for the case of stable trap sets; the proof for the case of supported trap sets is symmetrical.
	
	The set \(S\) is a \(\subseteq\)-minimal stable trap set of \(P\) \\
	iff \(S\) is a \(\subseteq\)-minimal trap set of \(\tgst{P}\) \\
	iff \(S\) is a simple cycle of \(\tgst{P}\) by~\Cref{prop:Datalog-min-trap-set-cycle-TS} \\
	iff \(S\) is a strict stable class of \(P\) by~\Cref{prop:Datalog-strict-StC-cycle-StTS}.
\end{proof}

Building on the previous characterizations, we can now establish that every stable or supported trap space necessarily covers at least one strict class of the corresponding type.

\begin{corollary}
  \label{cor:Datalog-TS-cover-strict-C}
	Let \(P\) be a \pname.
	Then every stable (resp.\ supported) trap space of \(P\) contains at least one strict stable (resp.\ supported) class of \(P\).
\end{corollary}
\begin{proof}
	Let \(I\) be a stable (resp.\ supported) trap space of \(P\).
	By definition, \(\cset{I}\) is a stable (resp.\ supported) trap set of \(P\).
	There is a \(\subseteq\)-minimal stable (resp.\ supported) trap set \(S\) of \(P\) such that \(S \subseteq \cset{I}\).
	By~\Cref{prop:Datalog-min-trap-set-strict-C}, \(S\) is a strict stable (resp.\ supported) class of \(P\).
	Now, we can conclude the proof.
\end{proof}

\Cref{cor:Datalog-TS-cover-strict-C} shows that a stable (resp.\ supported) trap space always covers at least one strict stable (resp.\ supported) class.
\Cref{prop:Datalog-two-min-s-TS-inconsistent} shows that two \(\leq_s\)-minimal stable (resp.\ supported) trap spaces are not consistent.
This implies that the number of \(\leq_s\)-minimal stable (resp.\ supported) trap spaces is a lower bound for the number of strict stable (resp.\ supported) classes in a \pname.
This insight is similar to the insight in \acbns that the number of minimal trap spaces of a \acbn is a lower bound for the number of attractors of this \acbn regardless of the employed update scheme~\citep{KBS2015}.

\subsubsection{Stable and Supported Partial Model Semantics}

Stable and supported trap spaces are also closely related to partial model semantics, particularly in the setting of \acbns. 
The following proposition, originally formulated in the context of \acbns, characterizes trap spaces via a natural order-theoretic condition on three-valued interpretations.

\begin{proposition}[Proposition 2 of~\cite{TBR2025}]\label{prop:BN-char-TS}
	Let \(f\) be a \acbn.
	A sub-space \(m\) is a trap space of \(f\) iff \(m(f_v) \leq_s m(v)\) for every \(v \in \var{f}\).
\end{proposition}

Building on this characterization, we now relate the supported trap spaces of a \pname to the trap spaces of its corresponding \acbn.

\begin{theorem}\label{theo:Datalog-SuTS-BN-TS}
	Let \(P\) be a \pname and \(f\) be its encoded \acbn\@.
	Then the supported trap spaces of \(P\) coincide with the trap spaces of \(f\).
\end{theorem}
\begin{proof}
	It is sufficient to show that the supported transition graph of \(P\) is identical to the synchronous state transition graph of \(f\).
	
	By definition, \(\var{f} = \hb{P}\), thus \(V(\tgsp{P}) = V(\sstg{f})\).
	Let \(I\) and \(J\) be two two-valued interpretations of \(P\).
	They are states of \(f\) as well.
	We have \((I, J) \in E(\tgsp{P})\) \\
	iff \(J = T_{P}(I)\) \\
	iff \(J(v) = I(\rhs{v})\) for every \(v \in \hb{P}\) \\
	iff \(J(v) = I(f_v)\) for every \(v \in \var{f}\) \\
	iff \((I, J) \in E(\sstg{f})\).
	This implies that \(E(\tgsp{P}) = E(\sstg{f})\).
\end{proof}

These above results immediately lead us to the model-theoretic characterization of supported trap spaces in \pnames.

\begin{corollary}\label{cor:Datalog-char-SuTS}
	Let \(P\) be a \pname.
	Then a three-valued interpretation \(I\) is a supported trap space of \(P\) iff \(m(\rhs{v}) \leq_s m(v)\) for every \(v \in \hb{P}\).
\end{corollary}
\begin{proof}
	This immediately follows from~\Cref{theo:Datalog-SuTS-BN-TS} and~\Cref{prop:BN-char-TS}.
\end{proof}

\Cref{cor:Datalog-char-SuTS} shows that supported trap spaces can be characterized in another way that is model-theoretic.
This also turns out that a supported trap space may not be a three-valued model of \(P\) as it considers the order \(\leq_s\), whereas the latter considers the order \(\leq_t\).
In~\Cref{exam:Datalog-trap-spaces}, \(I_2 = \{p = \fval, q = \tval, r = \uval\}\) is a supported trap space, but it is not a three-valued model of \(P\).

This observation highlights a subtle distinction between supported trap spaces and three-valued models, motivating the need to understand their relationships with other semantics.
To that end, we now show that every supported partial model of a \pname is indeed a supported trap space, thereby establishing a one-way implication between these two notions.

\begin{corollary}\label{cor:Datalog-SuPM-is-SuTS}
	Let \(P\) be a \pname.
	If \(I\) is a supported partial model of \(P\), then it is also a supported trap space of \(P\).
\end{corollary}
\begin{proof}
	Assume that \(I\) is a supported partial model of \(P\).
	Let \(f\) be the encoded \acbn of \(P\).
	By~\Cref{theo:Datalog-SuPM-BN-CoTS}, \(I\) is a complete trap space of \(f\).
	Then \(I\) is a trap space of \(f\) by~\Cref{prop:BN-char-TS}.
	By~\Cref{theo:Datalog-SuTS-BN-TS}, \(I\) is a supported trap space of \(P\).
\end{proof}

We next turn our attention to the stable case and show an analogous result: every stable partial model of a \pname is also a stable trap space.

\begin{proposition}\label{prop:Datalog-StPM-is-StTS}
	Let \(P\) be a \pname.
	If \(I\) is a stable partial model of \(P\), then it is also a stable trap space of \(P\).
\end{proposition}
\begin{proof}
	Assume that \(I\) is a stable partial model of \(P\).
	Let \(\lfp{P}\) be the least fixpoint of \(P\). \\
	By~\Cref{theo:Datalog-lfp-model-equivalence}, \(I\) is a stable partial model of \(\lfp{P}\). \\
	By~\Cref{cor:neg-Datalog-StPM-SuPM}, \(I\) is a supported partial model of \(\lfp{P}\) since \(\lfp{P}\) is negative. \\
	By~\Cref{cor:Datalog-SuPM-is-SuTS}, \(I\) is a supported trap space of \(\lfp{P}\). \\
	By~\Cref{cor:neg-Datalog-StTS-SuTS}, \(I\) is a stable trap space of \(\lfp{P}\) since \(\lfp{P}\) is negative. \\
	By~\Cref{cor:Datalog-lfp-StTS}, \(I\) is a stable trap space of \(P\).
\end{proof}

Having established that every supported partial model of a \pname is also a supported trap space, we now examine the converse direction.
Specifically, we show that every supported trap space contains (\wrttext \(\leq_s\)) some supported partial model, thereby revealing a form of approximation from below.
This leads to a further refinement: the notion of \(\leq_s\)-minimality coincide for supported partial models and supported trap spaces.
Together, these results establish a tight correspondence between the two notions in the minimal case.

\begin{corollary}\label{cor:Datalog-SuPM-leq-s-SuTS}
	Let \(P\) be a \pname.
	Then for every supported trap space \(I\) of \(P\), there is a supported partial model \(I'\) of \(P\) such that \(I' \leq_s I\).
\end{corollary}
\begin{proof}
	Let \(f\) be the encoded \acbn of \(P\). \\
	By~\Cref{theo:Datalog-SuTS-BN-TS}, \(I\) is a trap space of \(f\). \\
	By~\Cref{lem:BN-TS-inclusion-CoTS}, there is a complete trap space \(I'\) of \(f\) such that \(I' \leq_s I\). \\
	By~\Cref{theo:Datalog-SuPM-BN-CoTS}, \(I'\) is a supported partial model of \(P\).
\end{proof}

\begin{corollary}\label{cor:Datalog-min-s-SuTS-min-s-SuPM}
	Consider a \pname \(P\).
	Then a three-valued interpretation \(I\) is a \(\leq_s\)-minimal supported partial model of \(P\) iff \(I\) is a \(\leq_s\)-minimal supported trap space of \(P\).
\end{corollary}
\begin{proof}
	Let \(f\) be the encoded \acbn of \(P\). \\
	We have \(I\) is a \(\leq_s\)-minimal supported partial model of \(P\) \\
	iff \(I\) a \(\leq_s\)-minimal trap space of \(f\) by~\Cref{cor:Datalog-min-SuPM-BN-min-TS} \\
	iff \(I\) a \(\leq_s\)-minimal supported trap space of \(P\) by~\Cref{theo:Datalog-SuTS-BN-TS}.
\end{proof}

We now revisit the running example to illustrate the interplay between supported (stable) trap spaces and supported (stable) partial models.
This concrete instance not only highlights the relationships previously established but also provides a direct validation of several key results.
In particular, we demonstrate how supported partial models approximate supported trap spaces from below, and how minimality is preserved across the two notions.

\begin{example}\label{exam:Datalog-TS-PM}
	Let us continue with Example~\ref{exam:Datalog-trap-spaces}.
	The \pname \(P\) has five stable (resp.\ supported) trap spaces: \(I_1\), \(I_2\), \(I_3\), \(I_4\),  and \(I_5\).
	It has three stable (also supported) partial models: \(I_3\), \(I_4\),  and \(I_5\).
	This confirms the correctness of~\Cref{prop:Datalog-StPM-is-StTS} (resp.\ \Cref{cor:Datalog-SuPM-is-SuTS}).
	We have that \(I_4 \leq_s I_1\) and \(I_5 \leq_s I_2\), which confirms the correctness of~\Cref{cor:Datalog-SuPM-leq-s-SuTS}.
	The program \(P\) has two \(\leq_s\)-minimal supported partial models, namely \(I_4\) and \(I_5\), which are also two \(\leq_s\)-minimal supported trap spaces of \(P\).
	This is consistent with~\Cref{cor:Datalog-min-s-SuTS-min-s-SuPM}.
\end{example}

\subsubsection{Regular Model Semantics}

The regular model semantics not only inherits the advantages of the stable partial model semantics but also imposes two notable principles in non-monotonic reasoning: \emph{minimal undefinedness} and \emph{justifiability} (which is closely related to the concept of labeling-based justification in Doyle's truth maintenance system~\citep{D1979}), making it become one of the well-known semantics in logic programming~\citep{YY1994,JNSSY2006}.
Furthermore, regular models in ground \pnames were proven to correspond to preferred extensions in Dung's frameworks~\citep{WCG2009} and assumption-based argumentation~\citep{CS2017}, which are two central focuses in abstract argumentation~\citep{BCG2011}.

In~\Cref{exam:Datalog-TS-PM}, \(I_1\) and \(I_2\) are stable trap spaces, but they are not stable partial models of \(P\).
However, we observed that \(I_4\) and \(I_5\) are \(\leq_s\)-minimal stable trap spaces of \(P\).
They are also regular models of \(P\), i.e., the set of regular models of \(P\) coincides with the set of \(\leq_s\)-minimal stable trap spaces of \(P\), which is really interesting.
This observation can be generalized as follows:

\begin{theorem}[\textbf{main result}]\label{theo:Datalog-RegM-min-s-StTS}
	Let \(P\) be a \pname.
	Then a three-valued interpretation \(I\) is a regular model of \(P\) iff \(I\) is a \(\leq_s\)-minimal stable trap space of \(P\).
\end{theorem}
\begin{proof}
	Let \(\lfp{P}\) be the least fixpoint of \(P\). \\
	We have \(I\) is a regular model of \(P\) \\
	iff \(I\) is a regular model of \(\lfp{P}\) by~\Cref{theo:Datalog-lfp-model-equivalence} \\
	iff \(I\) is a \(\leq_s\)-minimal stable partial model of \(\lfp{P}\) by definition \\
	iff \(I\) is a \(\leq_s\)-minimal supported partial model of \(\lfp{P}\) by~\Cref{cor:neg-Datalog-StPM-SuPM} and the fact that \(\lfp{P}\) is negative \\
	iff \(I\) is a \(\leq_s\)-minimal supported trap space of \(\lfp{P}\) by~\Cref{cor:Datalog-min-s-SuTS-min-s-SuPM} \\
	iff \(I\) is a \(\leq_s\)-minimal stable trap space of \(\lfp{P}\) by~\Cref{cor:neg-Datalog-StTS-SuTS} \\
	iff \(I\) is a \(\leq_s\)-minimal stable trap space of \(P\) by~\Cref{cor:Datalog-lfp-StTS}.
\end{proof}

By~\Cref{prop:Datalog-two-min-s-TS-inconsistent} and~\Cref{theo:Datalog-RegM-min-s-StTS}, we deduce that two distinct regular models are separated.
Plus~\Cref{cor:Datalog-TS-cover-strict-C}, we deduce that the number of regular models of \(P\) is a lower bound for the number of strict stable classes of \(P\).
To the best of our knowledge, both insights are new in relationships between regular models and stable classes.
In addition, \Cref{prop:Datalog-exist-StTS-SuTS} implies that every \pname \(P\) has at least one regular model.
Note that the proof of the existence of a regular model in an \nlp relies on that of a stable partial model~\citep{YY1994}.

\subsection{Discussion}

In summary, we can first conclude that the notion of stable or supported trap space is a natural extension of the stable (supported) model semantics and the stable (supported) partial model semantics.
It is easy to see that for any \pname, the set of two-valued stable (resp.\ supported) trap spaces coincides with the set of stable (resp.\ supported) models.
By~\Cref{prop:Datalog-StPM-is-StTS} (resp.\ \Cref{cor:Datalog-SuPM-is-SuTS}), a stable (resp.\ supported) partial model is also a stable (resp.\ supported) trap space.

Second, the notion of stable trap space can also be viewed as an intermediate between model-theoretic semantics (the regular model semantics) and dynamical semantics (the stable class semantics).
The regular model semantics somewhat generalizes the main other model-theoretic semantics for \pnames, namely the stable model semantics, the well-founded model semantics, and the stable partial model semantics~\citep{YY1994,P1994}.
It also imposes the principle of minimal undefinedness, i.e., the undefined value should be used only when it is necessary~\citep{YY1994}.
By~\Cref{prop:Datalog-StPM-is-StTS}, the set of stable trap spaces includes the set of stable partial models, and thus also includes the set of regular models.
In addition, by~\Cref{theo:Datalog-RegM-min-s-StTS}, the set of \(\leq_s\)-minimal stable trap spaces coincides with the set of regular models.
Hence, the trap space semantics possesses both the model-theoretic aspect and the principle of minimal undefinedness inherent in the regular model semantics.
The stable class semantics expresses the dynamical aspect of a \pname~\citep{BS1992,BS1993}.
It is also characterized by the stable transition graph of the program~\citep{BS1992}.
The notion of stable trap space is defined based on the stable transition graph of a \pname as well.
Note that a stable class may not be a stable trap space due to the requirement for three-valued interpretations.
However, by~\Cref{cor:Datalog-TS-cover-strict-C}, we know that a stable (resp.\ supported) trap space contains at least one strict stable class, which represents a minimal oscillation between two-valued interpretations, and all the meaningful stable classes of a \pname are strict~\citep{IS2012}.
In particular, the notion of stable trap space reveals a deeper relationship between the regular model semantics and the stable class semantics, namely, that a regular model covers at least one strict stable class and the number of regular models is a lower bound for the number of strict stable classes of a \pname. 

Third, the relationships between \pnames and abstract argumentation have been deeply studied~\citep{Dung1995,CG2009,WCG2009,CSAD2015,CS2017,ASA2019}.
Abstract Argumentation Frameworks (AFs) are the most prominent formalisim for formal argumentation research~\citep{Dung1995,BCG2011}.
Abstract Dialectical Frameworks (ADFs) are more general than AFs, and have attracted much attention~\citep{BCG2011}.
However, some extension-based semantics that exist in AFs or ADFs do not have corresponding counterparts in \pnames.
The new notion of stable or supported trap space helps us to fill this gap.
It has been shown that the admissible sets of an AF correspond to the trap spaces of the respective \acbn~\citep{DDK2024,TBR2025}, and thus by~\Cref{theo:Datalog-SuTS-BN-TS}, the admissible sets of an AF correspond to the supported trap spaces of the respective \pname.
It has been shown that the admissible interpretations of an ADF coincide with the trap spaces of the respective \acbn~\citep{AGRZ2024,HKL2024}, and thus by~\Cref{theo:Datalog-SuTS-BN-TS}, the admissible interpretations of an ADF coincide with the supported trap spaces of the respective \pname.

\section{Conclusion}\label{sec:conclusion}

In this paper, we have established a formal link between \pnames and Boolean network theory in terms of both semantics and structure.
This connection has enabled us to import key concepts and results from the study of discrete dynamical systems into the theory and analysis of \pnames.

By analyzing the atom dependency graph of a \pname, we have identified structural conditions—specifically, the absence of odd or even cycles—that guarantee desirable semantic properties.
In particular, we have proved that:
(i) in the absence of odd cycles, the regular models coincide with the stable models, ensuring their existence;
(ii) in the absence of even cycles, the stable partial models are unique, which entails the \unitext of regular models.
Key to our proofs is the established connection and the existing graphical analysis results in Boolean network theory.
We have also revisited earlier claims made by~\cite{YY1994} regarding (i) and the regular model part of (ii) in normal logic programs.
While their intuition was partially correct, we have identified issues in their formal definition of well-founded stratification.
We have then clarified the scope of applicability to negative normal logic programs, thereby refining the theoretical landscape.

Beyond these structural insights, we have introduced several upper bounds on the number of stable, stable partial, and regular models based on the cardinality of a feedback vertex set in the atom dependency graph of a \pname. 
This provides a novel complexity measure grounded in graph-theoretic properties of \pnames.

Furthermore, we have obtained several stronger graphical analysis results on a subclass of \pnames, namely uni-rule \pnames~\citep{SS1997,CSAD2015}, which are important in the theory of \pnames, as well as being closely related to abstract argumentation frameworks~\citep{CSAD2015}.
These stronger results rely on the notion of delocalizing triple in signed directed graphs~\citep{RR2013}.

Finally, our investigation has led to a conceptual enrichment of \pnames through the notions of stable and supported trap spaces, borrowed from the notion of trap space in Boolean network theory.
We have formalized supported and stable trap spaces in both model-theoretic and dynamical settings, shown their basic properties, and demonstrated their relationships to other existing semantics, in particular, shown that the \(\leq_s\)-minimal stable trap spaces coincide with the regular models.
This correspondence offers a new perspective on the dynamics of \pnames and may open the door to new algorithms for the enumeration of canonical models.

\bibliography{main}

\end{document}